\documentclass[a4paper]{article}

\usepackage{fullpage}
\newcommand{\NLIPICS}{}

\usepackage{amsmath,amsthm,nicefrac}
\usepackage{amsfonts, amstext}
\usepackage[font={small,it}]{caption}

\usepackage{setspace}

\usepackage{enumitem}
\usepackage{xparse}

\usepackage[nameinlink]{cleveref}
\Crefname{algocf}{Algorithm}{Algorithms}
\crefname{algocfline}{line}{lines}
\Crefname{invariant}{Invariant}{Invariants}
\Crefname{claim}{Claim}{Claims}
\Crefname{observation}{Observation}{Observations}
\Crefname{subclaim}{Subclaim}{Subclaims}
\renewcommand{\paragraph}{\subsection}
\usepackage{epsfig}
\usepackage{amsthm,amssymb,mathrsfs}
\usepackage{xspace}
\usepackage{soul}
\usepackage{latexsym}
\usepackage{bbm}

\usepackage{framed}

\usepackage[dvipsnames]{xcolor}
\definecolor{DarkGray}{rgb}{0.66, 0.66, 0.66}
\definecolor{DarkPowderBlue}{rgb}{0.0, 0.2, 0.6}
\definecolor{fluorescentyellow}{rgb}{0.8, 1.0, 0.0}

\usepackage[ruled,vlined,linesnumbered,algonl]{algorithm2e}
\SetEndCharOfAlgoLine{}
\SetKwComment{Comment}{\footnotesize$\triangleright$\ }{}

\SetCommentSty{mycommfont}

\usepackage{thmtools,thm-restate}

\allowdisplaybreaks

\newcommand{\maxmin}{{\sc MaxMin}\xspace}
\newcommand{\minmax}{{\sc MinMax}\xspace}

\newcommand{\lrat}[1]{\nicefrac{\elx^{(#1)}}{\eln^{(#1)}}}

\newcommand{\alert}[1]{{\color{red}#1}}

\sethlcolor{fluorescentyellow}

\newcommand{\initOneLiners}{%
    \setlength{\itemsep}{0pt}
    \setlength{\parsep }{0pt}
    \setlength{\topsep }{0pt}
}

\usepackage{amsthm}

\ifdefined\NLIPICS
\newtheorem{theorem}{Theorem}[section]
\newtheorem{lemma}[theorem]{Lemma}

\newtheorem{remark}[theorem]{Remark}
\newtheorem{corollary}[theorem]{Corollary}

\newtheorem{observation}[theorem]{Observation}

\theoremstyle{definition}

\renewcommand{\theinvariant}{(I\@arabic\c@invariant)}
\fi
\newtheorem{fact}[theorem]{Fact}


\newcommand{\e}{\epsilon}

\newcommand{\Ldiff}{\Delta}
\newcommand{\cX}{\mathcal{X}}

\newcommand{\xmin}{x_{\min}}

\newcommand{\ep}{{\sc EP}}
\newcommand{\gp}{{\sc GP}}

\newcommand{\poly}{\operatorname{poly}}
\newcommand{\argmin}{\operatorname{argmin}}

\newcommand{\dsum}{\displaystyle\sum}

\newcommand{\junk}[1]{}
\newcommand{\eat}[1]{}

\usepackage{comment}
\usepackage{bm}
\usepackage{todonotes}
\usepackage{float}

\excludecomment{noappendix}

\title{A General Framework for Learning-Augmented Online Allocation\thanks{An extended abstract of this paper will appear in the Proceedings of the 50th EATCS International Colloquium on Automata, Languages and Programming (ICALP 2023). IC was supported in part by ISF grant 1737/21. DP was supported in part by NSF grants CCF-1750140 (CAREER Award) and CCF-1955703.}}

\author{Ilan Reuven Cohen\thanks{Faculty of Engineering, Bar-Ilan University, Israel. ilan-reuven.cohen@biu.ac.il}
\and{Debmalya Panigrahi \thanks{
Department of Computer Science, Duke University, Durham, NC, USA.
debmalya@cs.duke.edu}}}

\usepackage[bibliography=common]{apxproof}
\usepackage{verbatim}
\renewenvironment{toappendix}{}{}

\begin{document}
\date{}

\newcommand{\Rnn}{\R_{\geq0}}
\newcommand{\R}{\mathbb{R}}
\newcommand{\bw}{\mathbf{w}}
\newcommand{\bg}{\mathbf{g}}
\newcommand{\bp}{\mathbf{p}}
\newcommand{\bu}{\mathbf{u}}
\newcommand{\bep}{\mathbf{\epsilon}}
\newcommand{\Rp}{\R_{> 0}}
\newcommand{\mps}{\Rp^{m \times n}}
\newcommand{\vps}{\Rp^{m}}
\newcommand{\vpsn}{\Rp^{n}}

\newcommand{\lsan}{\ell^{\textbf{SNT}}}
\newcommand{\lmks}{\ell^{\textbf{MKS}}}

\newcommand*\bell{\ensuremath{\boldsymbol\ell}}
\newcommand{\elx}{\ell_{\max}}
\newcommand{\eln}{\ell_{\min}}
\newcommand{\sumds}{\displaystyle\sum}
\newcommand{\sumdsj}{\displaystyle\sum_{j\in[n]}}
\newcommand{\fracl}[2]{{#1}\big/{(#2)}}

\maketitle

\pagenumbering{gobble}

\begin{abstract}
    Online allocation is a broad class of problems where items arriving online have to be allocated to agents who have a fixed utility/cost for each assigned item so to maximize/minimize some objective. This framework captures a broad range of fundamental problems such as the Santa Claus problem (maximizing minimum utility), Nash welfare maximization (maximizing geometric mean of utilities), makespan minimization (minimizing maximum cost), minimization of  $\ell_p$-norms, and so on. We focus on divisible items (i.e., fractional allocations) in this paper. Even for divisible items, these problems are characterized by strong super-constant lower bounds in the classical worst-case online model. 
    
    In this paper, we study online allocations in the {\em learning-augmented} setting, i.e., where the algorithm has access to some additional (machine-learned) information about the problem instance. We introduce a {\em general} algorithmic framework for learning-augmented online allocation that produces nearly optimal solutions for this broad range of maximization and minimization objectives using only a single learned parameter for every agent. As corollaries of our general framework, we improve prior results of Lattanzi et al. (SODA 2020) and Li and Xian (ICML 2021) for learning-augmented makespan minimization, and obtain the first learning-augmented nearly-optimal algorithms for the other objectives such as Santa Claus, Nash welfare, $\ell_p$-minimization, etc. We also give tight bounds on the resilience of our algorithms to errors in the learned parameters, and study the learnability of these parameters.
\end{abstract}

\clearpage

\pagenumbering{arabic}

\section{Introduction}\label{sec:introduction}
Recent research has focused on obtaining learning-augmented algorithms for many online problems to overcome pessimistic lower bounds in competitive analysis. In this paper, we consider the {\em online allocation} framework in the learning-augmented setting. In this framework, a set of (divisible) items have to be allocated online among a set of agents, where each agent has a non-negative utility/cost for each item. This framework captures a broad range of classic problems depending on the objective one seeks to optimize. In load balancing (also called {\em makespan minimization}), the goal is to {\em minimize the maximum} (\minmax) cost of any agent. A more general goal is to minimize the $\ell_p$-norm of the cost vector defined on the agents, for some $p \ge 1$. Both makespan minimization (which is $\ell_\infty$-minimization) and $\ell_p$-minimization are classic problems in scheduling theory and have been extensively studied in competitive analysis. In a different vein, the online allocation framework also applies to maximization problems, where the allocation of an item obtains some utility for the receiving agent. This includes the famous Santa Claus problem, where the goal is to {\em maximize the minimum} (\maxmin) utility of any agent, or the maximization of {\em Nash welfare} which is defined as the geometric mean of the agents' utilities. These maximization objectives have also been been extensively studied, particularly because of their connection to {\em fairness} in allocations.

\smallskip\noindent{\bf Learning-Augmented Online Allocation.}
In this paper, we consider the online allocation framework in the {\em learning-augmented} setting. Typically, online allocation problems are characterized by strong super-constant lower bounds in competitive analysis, e.g., $\Omega(\log m)$ for load balancing~\cite{AzarNR95}, $\Omega(p)$ for $\ell_p$-minimization~\cite{AwerbuchAGKKV95} and $\Omega(m)$ for both Santa Claus (folklore) and Nash welfare~\cite{BanerjeeGGJ22}. A natural question, then, is whether some additional (machine-learned) information about the problem instance (we call these {\em learned parameters}) can help overcome these lower bounds and obtain a near-optimal solution. In this paper, we answer this question in the affirmative. In particular, we give a simple, unified framework for obtaining near-optimal (fractional) allocations {\em using a single learned parameter for every agent}. Our result holds for both maximization and minimization problems, and applies to all objective functions that satisfy two mild technical conditions that we define below. Indeed, the most interesting  aspect of our techniques and results is this generality: prior work for online allocation problems, both in {\em competitive analysis} and {\em beyond worst-case algorithms}, has typically been specific to the objective at hand, and the techniques for maximization and minimization objectives bear no similarity. In contrast, our techniques surprisingly handles not only a broad range of objectives but applies both to maximization and minimization problems simultaneously. We hope that the generality of our methods will cast a new light on what is one of the most important classes of problems in combinatorial optimization.

Before proceeding further, we define the two technical conditions that the objective function of the online allocation problem needs to satisfy for our results to apply. Let $f:\vps\rightarrow \Rp$ be the objective function defined on the vector of costs/utilities of the agents. Then, the conditions are:
\begin{itemize}
    \item {\em Monotonicity:} 
    $f$ is said to be {\em monotone} if the following holds: for any $\ell, \ell' \in \vps$ such that $\ell_i \ge \ell'_i$ for all $i\in [m]$, we have $f(\ell) \ge f(\ell')$.
    \item {\em Homogeneity:} 
    $f$ is said to be {\em homogeneous} if the following holds: for any $\ell, \ell' \in \vps$ such that $\ell'_i = \alpha \cdot \ell_i$ for all $i\in [m]$, then we have 
    $f(\ell') = \alpha \cdot f(\ell)$. 
\end{itemize}
We say an objective function is {\em well-behaved} if it is both monotone and homogeneous. All online allocation objectives studied previously that we are aware of are well-behaved, including the examples given above. 

\subsection{Our Results}

We now state our main result below:
\begin{theorem}[Informal]\label{thm:main-informal}
    Fix any $\epsilon > 0$. For any online allocation problem with a well-behaved objective, there is an algorithm that achieves a competitive ratio of  $1-\epsilon$ for maximization problems or $1+\epsilon$ for minimization problems using a single learned parameter for every agent.
\end{theorem}
We remark that the role of $\epsilon$ in the above theorem is to ensure that the learned parameter vector is of bounded precision.

\smallskip\noindent{\bf Comparison to Prior Work.} 
Lattanzi {\em et al.}~\cite{LattanziLMV20} were the first to consider online allocation in a learning-augmented setting. They considered a special case of the load balancing problem called restricted assignment, and showed the surprising result that a single (learned) parameter for each agent is sufficient to bypass the lower bound and obtain a nearly optimal (fractional) allocation. This result was further generalized by Li and Xian~\cite{LiX21} to the full generality of the load balancing problem, but instead of a single parameter, they now required two parameters for every agent. At a high level, their algorithm first uses one set of parameters to restrict the set of agents who can receive an item, and then solves the resulting restricted assignment problem using the second set of parameters. A a corollary of \Cref{thm:main-informal}, we improve this result by obtaining a near-optimal solution using a single learned parameter for every agent. In both these papers, as well as in our paper, the (fractional) allocation uses \emph{proportional allocation}. In the setting of online optimization, proportional allocations were used earlier by Agrawal {\em et al.}~\cite{agrawal2018proportional} for the (weighted) $b$-matching problem. As in our paper, they also gave an iterative algorithm for computing the parameters of the allocation. However, because the two problems are structurally very different (e.g., matching is a packing problem while our allocation problems are covering problems), the iterative algorithm in the Agrawal {\em et al.} paper is different from ours. To the best of our knowledge, our results for the other problems, namely Santa Claus, Nash welfare maximization, $\ell_p$-norm minimization, and other objectives that can be defined in the online allocation framework are the first results in learning-augmented algorithms for these problems.

\smallskip
We now state our additional results.

\smallskip\noindent{\bf Resilience to Prediction Error.} 
A key desiderata of learning-augmented online algorithms is resilience to errors in the learned parameters. In other words, one desires that the competitive ratio of the algorithm should gracefully degrade when the learned parameters used in the algorithm deviate from their optimal values. 
For well-behaved objectives for both minimization and maximization problems, we give an error-resilient algorithm whose competitive ratio degrades gracefully with prediction error:

\begin{theorem}[Informal]\label{noise}
    For any online allocation problem with a well-behaved objective, there is an (learning-augmented) algorithm that achieves a competitive ratio of $O(\alpha)$ when the learned parameter input to the algorithm is within a multiplicative factor of $\alpha$ of the optimal learned parameter for every agent. This holds for both minimization and maximization objectives.
\end{theorem}

The above theorem is asymptotically tight for the \maxmin objective. But, interestingly, for the \minmax objective we can do better: 
\begin{theorem}[Informal]\label{noise-min}
    For the load balancing problem (\minmax objective), there is an (learning-augmented) algorithm that achieves a competitive ratio of $O(\log \alpha)$ when the learned parameter input to the algorithm is within a multiplicative factor of $\alpha$ of the optimal learned parameter for every agent. Moreover, the dependence $O(\log \alpha)$ in the above statement is asymptotically tight.
\end{theorem}
An analogous statement was previously known only in the special case of restricted assignment~\cite{LattanziLMV20}.
\begin{remark} 
    We use a multiplicative measure of error $\alpha$ similar to  \cite{LattanziLMV20}.
For both \minmax and \maxmin objectives, we may assume w.l.o.g. that $\alpha \leq m$. This is because by standard techniques, it is possible to achieve $O(\min (\alpha,m))$ and $O(\log \min (\alpha,m))$ competitiveness for the \maxmin and \minmax objectives respectively. We also show that our bounds are asymptotically tight as a function of $\alpha$, in addition to matching existing lower bounds for the two problems as a function of $m$.
\end{remark}

\smallskip\noindent{\bf Learnability of Parameters.}
We also study the learnability of the parameters used in our algorithm. 
Following \cite{LiX21} and \cite{LavastidaMRX21a}, we adopt the PAC framework. We assume that each item is drawn independently (but not necessarily identically) from a distribution, and show a bound on the sample complexity of approximately learning the parameter vector under this setting. For the \maxmin and \minmax objectives, we show the following:
\begin{theorem}[Informal]\label{thm:learning-informal}
Fix any $\epsilon > 0$. For the online allocation problem with \maxmin or \minmax objectives, the sample complexity of learning a parameter vector that gives a $1-\epsilon$ (for \maxmin) or $1+\epsilon$ (for \minmax) approximation is $O(\frac{m}{\log m} \cdot \log \frac{m}{\e})$.
\end{theorem}
We note that a similar result was previously known for the \minmax objective (Li and Xian~\cite{LiX21}).
We also generalize this result to all well-behaved objectives subject to a technical condition of {\em superadditivity} for maximization or {\em subadditivity} for minimization. All the objectives described earlier in the introduction satisfy these conditions.

\paragraph{Our Techniques}
Our learning-augmented online algorithms for both minimization and maximization objectives follow from a single, unified algorithmic framework that we develop in this paper. This is quite surprising because in the worst-case setting, the online algorithms for the different objectives do not share any similarity (indeed have different competitive ratios), particularly between maximization and minimization problems. First, let us first consider the \minmax and \maxmin objectives. To use common terminology across these problems, let us call the cost/utility of an item $j$ to an agent $i$ the {\em weight} of item $j$ for agent $i$ and denote it $p_{i,j}$. Our common algorithmic framework uses proportional allocation according to the learned parameters of the agents. Let $w_i$ denote the parameter for agent $i$. Normally, proportional allocation would entail that we allocate a fraction $x_{i,j}$ of item $j$ to agent $i$ where $x_{i,j} = \frac{w_i p_{i, j}}{\sum_{i'} w_{i'} p_{{i'}, j}}$. But, this is clearly not adequate, since it would produce the same allocation for both the \maxmin and \minmax objectives. Specifically, if $p_{i, j}$ is {\em large} for a pair $i, j$, then $x_{i, j}$ should be large for the \maxmin objective and small for the \minmax objective respectively. To implement this intuition, we exponentiate the weight $p_{i, j}$ by a fixed value $\alpha$ that depends on the objective (i.e., is different for \maxmin and \minmax) and then allocate using fractions $x_{i,j} = \frac{w_i p_{i, j}^\alpha}{\sum_{i'} w_{i'} p_{{i'}, j}^\alpha}$. We call this an {\em exponentiated proportional} allocation (or \ep-allocation in short), and call $\alpha$ the {\em exponentiation constant}.

Let us fix any value of $\alpha$. It is clear that for both the \minmax and \maxmin objectives, an optimal allocation has {\em uniform} cumulative fractional weights (called {\em load}) across all agents. (Note that otherwise, an infinitesimal fraction of an item can be repeatedly moved from the most loaded to the least loaded agent to eventually improve the competitive ratio.) Following this intuition, we define a {\em canonical allocation} as one that sets learned parameters on the agents in a way that equalizes the loads on all agents. We show that the canonical allocation always exists and is {\em unique}. Indeed, this is true not only for all \ep-allocation algorithms, but for a much broader class of proportional allocation schemes that we called {\em generalized proportional} allocations (or \gp-allocations). In the latter class, we allow any transformation of the weights $p_{i, j}$ before applying proportional allocation. Thus, \ep-allocations represent the subclass of \gp-allocations where the transformation is exponentiation by the fixed value $\alpha$. We also give a simple iterative (Sinkhorn-like) algorithm for computing the optimal learned parameters, and establish its convergence properties, for \gp-allocations. 
\gp-allocations give an even larger palette of proportional allocation schemes to choose from than \ep-allocations, and we hope it will be useful in future work for problem settings that are not covered in this paper (e.g., non-linear utilities).

Finally, we need to set the value of $\alpha$ specifically for the \minmax and \maxmin objectives. Intuitively, it is clear that we need to set $\alpha$ to a large {\em positive} value for the \maxmin objective and a large {\em negative} value for the \minmax objective. Indeed, we show that in the limit of $\alpha\rightarrow\infty$ and $\alpha\rightarrow-\infty$, the canonical allocation defined above recovers optimal allocations for the \maxmin and \minmax objectives respectively. We also show a monotonicity property of the optimal objective (with the value of $\alpha$) that can be used to set $\alpha$ to a finite value (function of $\epsilon$) and obtain a $1-\epsilon$ (resp., $1+\epsilon$) approximation for the \maxmin (resp., \minmax) objective, for any $\epsilon > 0$.

Now that we have described the \ep-allocation scheme for obtaining nearly optimal algorithms for the \minmax and \maxmin objectives, we generalize to all well-behaved objective functions. This is quite simple. The main advantage of the \minmax and \maxmin objectives that is not shared by other objectives is the property that the optimal solution has uniform load across all agents. Now, suppose for a maximization objective, the load of agent $i$ in an optimal solution is $s_i$ (we call this the {\em scaling parameter} for agent $i$). For now, suppose these values $s_i$ are also provided offline as a second set of parameters. Then, we can first scale the weights $p_{i, j}$ using these parameters to obtain a new instance $q_{i, j} = \frac{p_{i,j}}{s_i}$. Clearly, the optimal solution for the original instance has uniform load across all agents for the transformed instance. Indeed, by the monotonicity of the maximization objective, this solution for the transformed instance is also optimal for the \maxmin objective. Using the above analysis for the \maxmin objective, we can now claim that there exist learned parameters $w_i$ for $i\in [m]$ such that setting $x_{i,j} = \frac{w_i q_{i, j}^\alpha}{\sum_{i'} w_{i'} q_{{i'}, j}^\alpha}$ gives an optimal solution to the original instance of the problem. Now, note that 
\[
    x_{i,j} 
    = \frac{w_i q_{i, j}^\alpha}{\sum_{i'} w_{i'} q_{{i'}, j}^\alpha}
    = \frac{(w_i/s_i^\alpha) p_{i, j}^\alpha}{\sum_{i'} (w_{i'}/s_{i'}^\alpha) p_{{i'}, j}^\alpha}
    = \frac{w'_i p_{i, j}^\alpha}{\sum_{i'} w'_{i'} p_{{i'}, j}^\alpha} 
    \text{ for } w'_i = w_i/s_i^\alpha.
\]    
It follows that by using learned parameters $w'_i$ in an \ep-allocation on the original instance, we can obtain an optimal solution for the original maximization objective. (The case for a minimization objective is identical to the above argument, with the \maxmin objective being replaced by the \minmax objective.) Finally, using the homogeneity of the objective function, we can also set $\alpha$ to a finite value (function of $\epsilon$) and obtain a $1-\epsilon$ (resp., $1+\epsilon$) approximation for the maximization (resp., minimization) objective, for any $\epsilon > 0$.



\subsection{Related Work}
Learning-augmented online algorithms were pioneered by the work of Lykouris and Vassilvikskii~\cite{LykourisV21} for the caching problem, and has become a very popular research area in the last few years. The basic idea of this framework is to augment an online algorithm with (machine-learned) predictions about the future, which helps overcome pessimistic worst case lower bounds in competitive analysis. Many online allocation problems have been considered in this framework in scheduling~\cite{PurohitSK18,AzarLT21,AzarLT22,BamasMRS20,ImKQP21,Mitzenmacher20}, online matching~\cite{AntoniadisGKK20,ChenI21,KumarPSSV19}, ad delivery~\cite{MahdianNS12,LavastidaMRX21b}, etc. The reader is referred to the survey by Mitzenmacher and Vassilvitskii~\cite{MitzenmacherV20,MitzenmacherV22} for further examples of online learning-augmented algorithms. The papers specifically related to our work are those of Lattanzi et al.~\cite{LattanziLMV20} and Li and Xian~\cite{LiX21} that we described above, and that of Lavastida et al.~\cite{LavastidaMRX21a} that focuses on the learnability of the parameters for the same problem.
As mentioned earlier, Agrawal {\em et al.}~\cite{agrawal2018proportional} used the proportional allocation framework earlier for the online (weighted) $b$-matching problem, and gave an iterative algorithm for computing the parameters of the allocation. 

We now give a brief summary of online allocation in the worst-case model. For minimization problems, two classic objectives are makespan (i.e., $\ell_\infty$ norm) and $\ell_p$ norm minimization for $p > 1$. The former was studied in several works (e.g.,~\cite{AzarNR95,AspnesAFPW97}), eventually leading to an asymptotically tight bound of $\Theta(\log m)$. This was later generalized to arbitrary $\ell_p$ norms, and a tight bound of $\Theta(p)$ was obtained for this case~\cite{AwerbuchAGKKV95,Caragiannis08}. For maximization objectives, there are $\Omega(m)$ lower bounds for many natural objectives such as \maxmin (see, e.g., \cite{HajiaghayiKPS22}) and Nash welfare~\cite{BanerjeeGGJ22}. Some recent work has focused on overcoming these lower bounds using additional information such as monopolist values for the agents~\cite{BanerjeeGGJ22,BarmanKM22}. While this improves the competitive ratio to sub-linear in $m$, lower bounds continue to rule out near-optimal solutions (or even constant factor approximations) that we seek in this paper. 

\eat{
\alert{General}
We say that an objective function $f:\vps\rightarrow \Rp$ is {\em monotone} if the following holds: for any $\ell, \ell' \in \vps$ such that $\ell_i \ge \ell'_i$ for all $i\in [m]$, we have $f(\ell) \ge f(\ell')$.

We say that an objective function $f:\vps\rightarrow \Rp$ is $s$-{\em homogeneous} a fixed $s \geq 1$, if the following holds: for any $\ell, \ell' \in \vps$ such that $\ell'_i = (1+\epsilon) \cdot \ell'_i$ for all $i\in [m]$ and a fixed $s \geq 1$, then we have $f( (1+\epsilon)\cdot \ell) = (1+\epsilon)^s \cdot f(\ell)$.
\alert{The definition of homogeneous changed}
\begin{theorem}
\label{thm:general}
    Fix any instance of an online allocation problem with divisible items where the goal is to maximize or minimize a monotone homogeneous objective function. Then, there exists an online algorithm and a learned parameter vector in $\vps$ that achieves a competitive ratio of $1-\epsilon$ (for maximization) or $1+\epsilon$ (for minimization).
\end{theorem}
}

\smallskip\noindent{\bf Organization.} 
For most of the paper, we only consider the \minmax and \maxmin objectives. We establish the notation in \Cref{sec:contributions} and give an overview of the results. Then, we prove these results by showing properties of \gp-allocations in \Cref{sec:canonical} and of \ep-allocations in \Cref{sec:ep-monotonicity-convergence}. 
Next, we give noise resilient algorithms in \Cref{sec:noise} and discuss learnability of the parameters in \Cref{sec:learning}. Finally, in \Cref{sec:general}, we extend our results to all well-behaved objective functions via simple reductions to the \maxmin and \minmax objectives.

\section{Preliminaries and Results}\label{sec:contributions}
\subsection{Problem Definition}
We have $n$ (divisible) items that arrive online and have to be (fractionally) allocated to $m$ agents. The weight of item $j\in [n]$ for agent $i\in [m]$ is denoted $p_{i,j}$ and is revealed when item $j$ arrives. We denote the {\em weight matrix}
\[
    P = 
    \begin{bmatrix}
    p_{1,1} & \ldots & p_{1,n}\\
    \vdots & \ddots & \vdots\\
    p_{m,1} & \ldots & p_{m,n}\\
    \end{bmatrix}
    \text{ where all $p_{i,j}> 0$ for all $i\in [m], j\in [n]$}.\footnote{For notational simplicity (e.g., avoid dividing by 0), we will assume that all weights are strictly positive instead of being nonnegative. In both problems, this is without loss of generality. In load balancing, an item with weight $0$ for an agent can be removed by assigning it to the agent. For Santa Claus, we can replace weight $0$ by an arbitrarily small positive weight $\delta > 0$.}
\]

A feasible allocation is given by an {\em assignment matrix}
\[
    X = 
    \begin{bmatrix}
    x_{1,1} & \ldots & x_{1,n}\\
    \vdots & \ddots & \vdots\\
    x_{m,1} & \ldots & x_{m,n}\\
    \end{bmatrix}
    \text{ where } x_{i,j} \in [0, 1] \text{ for all } i\in [m], j\in [n] \text{ and } \sum_{i=1}^m x_{i,j} = 1 \text{ for all } j\in [n].
\]

Note that every item has to be fully allocated among all the agents. 
We use $\cX$ to denote the set of feasible solutions.
The total weight of an agent $i$ corresponding to an allocation $X$ (we call this the {\em load} of $i$) is given by 
\[
    \ell_i(P,X) = \sum_{j\in [n]} x_{i,j}\cdot p_{i,j},
\]
and the vector of loads of all the agents is denoted $\bell(P,X)$.

The load balancing problem is now defined as
\[
    \min_{X\in \mathcal{X}} \Big\{T: \ell_i(P,X) \leq T \text{ for all } i\in [m]\Big\},
\]    
%
while the Santa Claus problem is defined as
\[
    \max_{X\in \mathcal{X}} \Big\{T: \ell_i(P,X) \geq T \text{ for all } i\in [m]\Big\}.
\]


\subsection{Exponentiated and Generalized Proportional Allocations}
Our algorithmic framework is simple: when allocating item $j$, we first exponentiate the weights $p_{i, j}$ to $p_{i, j}^\alpha$ for some fixed $\alpha$ (called the {\em exponentiation constant}) that only depends on the objective being optimized. Next, we perform proportional allocation weighted by the learned parameters $w_i$ for agents $i\in [m]$:
\[
    x_{i,j} = \frac{p_{i,j}^\alpha \cdot w_i}{\displaystyle\sum_{i'\in [m]}p_{i',j}^\alpha\cdot w_{i'}}.
\]
We call this an {\em exponentiated proportional} allocation or \ep-allocation in short.

Our main theorem is the following:
\begin{theorem}\label{thm:epa}
    For the load balancing and Santa Claus problems, there are \ep-allocations that achieve a competitive ratio of $1+\e$ and $1-\e$ respectively, for any $\e > 0$.
\end{theorem}

\smallskip\noindent{\bf The Canonical Allocation.} 
In order to define an \ep-allocation and establish \Cref{thm:epa}, we need to specify two things: the vector of learned parameters $\bw \in \vps$ and the exponentiation constant $\alpha$. First, we focus on the learned parameters. For any fixed $\alpha$ and a weight matrix $P$, we use learned parameters $\bw \in \vps$ that result in {\em equal load} for every agent. We call this the {\em canonical allocation}. The corresponding learned parameters and the load of every agent are respectively called the {\em canonical parameters} (denoted $\bw^*$) and the {\em canonical load} (denoted~$\ell^*$). 

Apriori, it is not clear that a canonical allocation should even exist, and even if it does, that it is unique. Interestingly, we show this existence and uniqueness not just from \ep-allocations but for the much broader class of proportional allocations where {\em any} function $f:\Rp\rightarrow\Rp$ (called the {\em transformation function}) can be used to transform the weights rather than just an exponential function. I.e.,
\[
    x_{i,j} = \frac{f(p_{i,j}) \cdot w_i}{\displaystyle\sum_{i'\in [m]}f(p_{i',j})\cdot w_{i'}}.
\]
We call this a {\em generalized proportional} allocation or \gp-allocation in short.

We show the following theorem for \gp-allocations:
\begin{theorem}\label{thm:maincanonical}
For any weight matrix $P \in \mps$ and any transformation function $f:\Rp\rightarrow\Rp$, the canonical load for a \gp-allocation exists and is unique. Moreover, it is attained by a unique (up to scaling) set of canonical parameters.
\end{theorem}

We prove \Cref{thm:maincanonical} algorithmically by giving a simple iterative (offline) algorithm that converges to the set of canonical parameters (see \Cref{alg:iterative}). 
We will show later that the canonical allocations produced by appropriately setting the value of the exponentiation constant $\alpha$ are respectively optimal (fractional) solutions for the Santa Claus and the load balancing problems. Therefore, an interesting consequence of the iterative convergence of this algorithm to the canonical allocation is that it gives a simple alternative {\em offline} algorithm for computing an optimal fractional solution for these two problems. To the best of our knowledge, this was not explicitly known before our work. 

An interesting direction for future research would be to explore other natural classes of transformation functions, other than the exponential functions considered in this paper. Since \Cref{thm:maincanonical} holds for any transformation function, they also admit a canonical allocation, and it is conceivable that such canonical allocations would optimize objective functions other than the \minmax and \maxmin functions considered here. For example, one natural open problem is following: are there a transformation functions whose canonical allocations correspond to maximizing Nash Social Welfare or minimizing $p$-norms of loads?

\smallskip\noindent{\bf Monotonicity and Convergence of \ep-allocations.}
Now that we have defined the learned parameters in \Cref{thm:epa} as the corresponding canonical parameters, we are left to define the values of the exponentiation constant $\alpha$ for the \maxmin and \minmax problems respectively. We show two key properties of canonical loads of \ep-allocations. First, we show that the canonical load is monotone nondecreasing with the value of $\alpha$. This immediately suggests that we should choose the largest possible value of $\alpha$ for the \maxmin problem since it is a maximization problem, and the smallest possible value of $\alpha$ for the \minmax problem since it is a minimization problem. Indeed, the second property that we show is that in the limit of $\alpha \rightarrow\infty$, the canonical load converges to the optimal objective for the Santa Claus problem (we denote this optimal value $\lsan$) and in the limit of $\alpha \rightarrow -\infty$, the canonical load converges to the optimal objective for the load balancing problem (we denote this optimal value $\lmks$). 

For a fixed $\alpha$, let $X(P,\alpha,\bw)$ denote the assignment matrix and $\bell(P,\alpha,\bw)$ the load vector for a learned parameter vector $\bw$. Let $\bell^*(P,\alpha)$ denote the corresponding canonical load. We show the following properties of canonical \ep-allocations:
\begin{theorem}
\label{thm:exponential} For any weight matrix $P \in \mps$, the following properties hold for canonical \ep-allocations:
\begin{itemize}
    \item The monotonicity property:
For $\alpha_1, \alpha_2 \in \R$ such that $\alpha_1 \geq \alpha_2$, we have $\ell^*(P,\alpha_1) \geq \ell^*(P,\alpha_2)$.
    \item The convergence property:
$\displaystyle\lim_{\alpha \rightarrow \infty} \ell^*(P,\alpha) = \lsan(P)$ and
$\displaystyle\lim_{\alpha \rightarrow -\infty} \ell^*(P,\alpha) = \lmks(P)$.
\end{itemize}
\end{theorem}

Clearly, \Cref{thm:exponential} implies \Cref{thm:epa} as a corollary when $\alpha$ is set sufficiently large for the Santa Claus problem and sufficiently small for the load balancing problem.

In the rest of the paper, we will prove \Cref{thm:maincanonical} and \Cref{thm:exponential}. 




\section{Canonical Properties of Generalized Proportional Allocations}\label{sec:canonical}
In this section, we prove Theorem~\ref{thm:maincanonical}. For notational convenience, we define a transformation matrix $G\in \mps$ where $G(i, j) = f(p_{i,j})$ for the transformation function $f$. Using this notation, we denote by $x_{i, j}(G, \bw)$ the fractional allocation of item $j$ to agent $i$, and by $\ell_i(P, G, \bw)$ the load of agent $i$ (we use $\bell(P, G, \bw)$ to denote the vector of agent loads) under the \gp-allocation corresponding to the transformation matrix $G$ and learned parameters $\bw$. 

We say two sets of learned parameters $\bw, \bw'$ are {\em equivalent} (denoted $\bw\equiv\bw'$) if there exists some constant $c > 0$ such that $w'_i = c\cdot w_i$ for every agent $i\in [m]$. 
The following is a simple observation from the \gp-allocation scheme that two equivalent sets of learned parameters produce the same allocation:
\begin{observation}
\label{obv:scale}
For any $G \in \mps$, if $\bw \equiv \bw' \in \Rp^m$,
then $x_{i,j}(G,\bw) = x_{i,j}(G,\bw') $ for all $i, j$.
\end{observation}

We also note that \gp-allocations are monotone in the sense that if one agent's parameter decreases while the rest increase, then the allocation on this agent decreases as well.
\begin{observation}
\label{obv:monotone}
Consider any $G \in \mps$ and any nonzero vector $\bm{\epsilon} \in \Rnn^m$ such that $-w_k < \epsilon_k \le 0$ for some $k\in [m]$ and $\epsilon_i\geq 0$ for all $i\ne k$. Then, $x_{k,j}(G,\bw') < x_{k,j}(G,\bw)$ for all $j\in [n]$, where $\bw' = \bw+\bm{\epsilon}$ and $\bw' \not= \bw$.
\end{observation}

Our first nontrivial property is that the load vector uniquely determines the learned parameters up to equivalence of the parameters.

\begin{lemma}
\label{lem:unique}
For any $P, G \in \mps$, $\ell_i(P, G,\bw) = \ell_i(P, G,\bw')$ for all $i\in [m]$ if and only if $\bw \equiv \bw'$.
\end{lemma}
\begin{proof}
In one direction, if $\bw\equiv\bw'$, the loads are identical because the allocations are identical (by \Cref{obv:scale}). 

We now show the lemma in the opposite direction.
Let $k = \arg \min_i \frac{w_i}{w'_i}$ and $c = \frac{w_k}{w'_k}$. Let us define $\hat{\bw} = c\cdot \bw'$. Then, $\hat{w}_k = w_k$, and $\hat{w}_{i'} =  \left(\min_i \frac{w_i}{w'_i}\right) \cdot w'_{i'} \leq \ w_{i'}$ for all $i'\not= k$. Now, if $\bw$ and $\bw'$ are not equivalent, then there exists some $i'\in [m]$ such that $\hat{w}_{i'} < w_{i'}$. Therefore, by Observation~\ref{obv:monotone}, $x_{k,j}(G,\hat{\bw}) > x_{k,j}(G,\bw)$ for all $j\in [n]$. But, by Observation~\ref{obv:scale}, $x_{k,j}(G,\hat{\bw}) = x_{k,j}(G,\bw')$ for all $j\in [n]$. Thus,  $x_{k,j}(G,\bw') > x_{k,j}(G,\bw)$ for all $j\in [n]$, which contradicts $\ell_k(P, G,\bw') = \ell_k(P, G,\bw)$.
\end{proof}

Similarly, we show that if the canonical load exists (i.e., a load vector where all loads are identical), it must be unique.

\begin{lemma}
\label{lem:uniq}
For any $P, G \in \mps$, if there exist $\bw,\bw'\in \vps$ such that $\ell_i(P, G,\bw) =\ell$ and $\ell_i(P, G,\bw') = \ell'$ for all $i\in [m]$, then $\ell = \ell'$.
\end{lemma}
\begin{proof}
Assume for the purpose of contradiction that there exist $\bw,\bw' \in \vps$  such that for all $i\in [m]$, $\ell_i(P,G,\bw) = \ell$ and $\ell_i(P,G,\bw') = \ell'$ but $\ell > \ell'$.
Let $k = \arg \min_i \frac{w_i}{w'_i}$ and $c = \frac{w_k}{w'_k}$, and let $\hat{\bw} = c\cdot \bw'$.
We have 
$$\ell' = \ell_k(P,G,\bw') = \ell_k(P,G,\hat{\bw}) \geq \ell_k(P,G,\bw) = \ell, \text{which is a contradiction.}$$ 
Here, the second equality is by Observation~\ref{obv:scale},
and the inequality is by Observation~\ref{obv:monotone}, since $\hat{w}_k = w_k$, and $\hat{w}_i \leq  w_i$ for $i\in [m]$. 
\end{proof}

\subsection{Convergence of \Cref{alg:iterative}}

The rest of this section focuses on showing the existence of a canonical allocation for \gp-allocations. We do so by showing convergence of the following simple iterative algorithm (\Cref{alg:iterative}):

\begin{algorithm}[ht]
\begin{itemize}
    \item Initialize: $\bw^{(0)} \leftarrow \mathbf{1}^m$ 
\end{itemize}
Iteration $r$:
    \begin{itemize}
    \item Compute $\bell^{(r)}$ as $\ell^{(r)}_i \leftarrow \ell_i(P,G,\bw^{(r)})$, for all $i\in[m]$, where $\ell_i(P,G,\bw^{(r)})$ is the load of agent $i$ under the \gp-allocation with transformation matrix $G$ and learned parameters $\bw^{(r)}$.
    \item Set $\bw^{(r+1)}$ as $w^{(r+1)}_i \leftarrow \frac{w^{(r)}_i}{\ell^{(r)}_i}\cdot \gamma^{(r)}$, for all $i\in[m]$.\\ 
    Here, $\gamma^{(r)} \in \R_{>0}$ is a scaling factor whose value does not affect the load (by \Cref{obv:scale}). But, by using, e.g., $\gamma^{(r)} = \ell^{(r)}_1$, we can ensure that the algorithm terminates with a single set of learned parameters instead of repeatedly finding equivalent sets of parameters after it has converged.
    
\end{itemize}
\caption{The iterative algorithm showing the existence of a canonical allocation for \gp-allocations.}
\label{alg:iterative}
\end{algorithm}

Note that \Cref{alg:iterative} ensures that if the loads of all agents are uniform at any stage, then the iterative process has converged and the algorithm terminates. So, it remains to show that for any $P,G \in \mps$, this iterative process reaches a set of parameters $\bw^* \in \vps$ such that $\ell_i(P,G,\bw^*) = \ell_{i'}(P,G,\bw^*)$ for all $i,i'\in [m]$. 

Our proof has two parts. 
The first part shows that the maximum and minimum loads are (weakly) monotone over the course of the iterative process. 
For this, we focus on a single iteration. For a vector $\bell \in \vps$, let $\ell_{\max} = \max_{i\in [m]} \ell_i$ and $\ell_{\min} = \min_{i\in [m]} \ell_i$ be the maximum and minimum coordinates of $\bell$.
We will show that if $\elx^{(r)}$ and $\eln^{(r)}$ are not equal at the beginning of an iteration, then $\elx^{(r)}$ can only decrease (or stay unchanged) and $\eln^{(r)}$ can only increase (or stay unchanged) in a single iteration. 
\begin{lemma}
\label{lem:notincreasing}
Consider any $P,G \in \mps$, $\gamma > 0$. Let $\bw,\bw',\bell,\bell' \in \vps$ such that $\ell_i = \ell_i(P,G,\bw)$, $\ell'_i = \ell_i(P,G,\bw')$ and $w'_i = \frac{w_i}{\ell_i}\cdot \gamma$ and let $\tilde{p}_i = 
\sum_j p_{i,j}$. Then,
we have
$\ell_i' \geq \nicefrac{\eln}{\left(1-\frac{\ell_i-\eln}{\tilde{p}_i}\right)}$ and
$\ell_i' \leq \nicefrac{\elx}{\left(1+\frac{\elx-\ell_i}{\tilde{p}_i}\right)}$
\end{lemma}

In the second part, we show that the ratio $\frac{\elx^{(r)}}{\eln^{(r)}}$ is strictly decreasing after a finite number of iterations. The proof of this stronger property requires the per-iteration weak monotonicity property that we establish in the first part of the proof. 
\begin{noappendix}
\end{noappendix}

\begin{lemma}
\label{lem:strictsingle}
Let $P, G\in \mps$ be given fixed matrices. Fix an iteration $r$ in \Cref{alg:iterative} where $\elx^{(r)} > \eln^{(r)}$. Let $\elx^{(r)} \geq (1+\epsilon)\cdot \eln^{(r)}$ for some $\epsilon \in (0, 1]$. Then, in the next iteration, we have $\eln^{(r+1)}  \geq (1+c \cdot \epsilon)\cdot \eln^{(r)} $ for some constant $c > 0$ that only depends on $P$ and $G$. 
\end{lemma}

Using \Cref{lem:notincreasing} and \Cref{lem:strictsingle}, we complete the proof of
\Cref{thm:maincanonical}.

\begin{proof}[Proof of \Cref{thm:maincanonical}]
\label{sec:proof-maincanonical}

We are given fixed matrices $P,G \in \mps$. Let $\elx^{(r)},\eln^{(r)}$ denote the maximum and the minimum load respectively in iteration $r$ of \Cref{alg:iterative}.
Let $c > 0$ be the constant (that depends only on $P, G$) in \Cref{lem:strictsingle}.

For a non-negative integer $a$, let $r_a$ be defined recursively as follows: 
\[
r_a = r_{a-1} + \left \lceil \frac{\log (1+2^{-a+1})}{\log{(1+c\cdot 2^{-a})}} \right \rceil + 1, 
\text{ where }
 r_0 = \left \lceil \frac{\log \left(\lrat{0}\right)}{\log{(1+c)}} \right \rceil + 1.
 \]
We will show for any $a$, in any iteration $r\geq r_a$, we have $\lrat{r} \leq 1+2^{-a}$.
First, we prove it for $a=0$.  If there exists some $r\leq r_0$ such that $\lrat{r} \leq 2$, 
then this also holds for $r\geq r_0$ by \Cref{lem:notincreasing}. 
Otherwise, for all $r\leq r_0$ we have $\lrat{r} > 2$. Then,
using \Cref{lem:strictsingle} with $\epsilon = 1$, we get  $\eln^{(r+1)}  \geq (1+c)\cdot \eln^{(r)}$.
Therefore,
$\eln^{(r_0)} \geq (1+c)^{r_0} \cdot \eln^{(0)} > \elx^{(0)}$ by our choice of $r_0$. 
This contradicts \Cref{lem:notincreasing}, thereby showing that $\lrat{r} \leq 2$
for any $r\ge r_0$. 

Now, we show the inductive case. Assume the inductive hypothesis that $\lrat{r_{a-1}} \leq 1+2^{-(a-1)}$. 
We will prove that $\lrat{r_{a}} \leq 1+2^{-(a)}$. The proof is similar to the base case of $a = 0$.
If there exists some $r\leq r_a$ such that $\lrat{r} \leq 1+2^{-a}$, then this inequality also holds for any $r\geq r_{a}$ by \Cref{lem:notincreasing}. 
Otherwise, for all $r\leq r_a$ we have $\lrat{r} > 1+2^{-a}$. 
Then,
for all $r_{a-1} \leq r\leq r_a$, using \Cref{lem:strictsingle} with $\epsilon = 2^{-a}$, we have  $\eln^{(r+1)}  \geq (1+c\cdot 2^{-a})\cdot \eln^{(r)}$. 
Therefore, $\eln^{(r_a)} \geq (1+c\cdot 2^{-a})^{r_a-r_{a-1}} \cdot \eln^{(r_{a-1})}$. 
By our choice of $r_a$, this implies 
$\eln^{(r_a)} > (1+2^{-(a-1)})\cdot \eln^{(r_{a-1})}$. By the induction hypothesis, this implies $\eln^{(r_a)} > \elx^{(r_{a-1})}$. But, this implies $\elx^{(r_a)} > \elx^{(r_{a-1})}$, which contradicts \Cref{lem:notincreasing}. Therefore,
$$\lim_{r\rightarrow \infty} \nicefrac{\elx^{(r)}}{\eln^{(r)}} = 1,$$ and $\ell^*(P,G) = \displaystyle\lim_{r\rightarrow \infty} \elx^{(r)}$. Moreover, by \Cref{lem:uniq} this value is uniquely defined and attained by a unique (up to scaling) set of learned parameters. 

\end{proof}

\subsection{Weak Monotonicity of the Maximum and Minimum Loads in \Cref{alg:iterative}: Proof of \Cref{lem:notincreasing}}

For ease of description, we assume that $G$ and $\bw$ are normalized in the following sense: 
\[
    \bw = \mathbf{1}^m \text{ and } \sum_{j} g_{i,j} = 1.
\]    
This transformation is local to the current iteration, and only for the purpose of this proof. First, we explain why this change of notation is w.l.o.g.
Suppose $\hat{G}, {\bf \hat{w}}$ represent the actual transformation matrix and learned parameters respectively. Now, we define $G$ as follows:
$$g_{i,j} = \frac{\hat{g}_{i,j} \cdot \hat{w}_i} {\sum_{i'\in [m]} \hat{g}_{i',j} \cdot \hat{w}_{i'}},$$
and our new learned parameters is given by $\mathbf{1}^m$.

Note that the fractional allocation remains unchanged, i.e., $x_{i,j}(\hat{G},\hat{w}) = x_{i,j}(G,\mathbf{1}^m) = g_{i,j}$, and therefore the loads are also unchanged:
    $\ell_i = \ell_i(P,\hat{G},\hat{\bw}) = 
\ell_i(P,G,\mathbf{1}^m) = \sum_{j\in[n]} g_{i,j} \cdot p_{i,j}$.
Assume w.l.o.g. (by Observation~\ref{obv:scale}) that $\gamma = \ell_1$, so $\hat{w}'_i =\frac{\hat{w}_i}{\ell_i}\cdot \ell_1$. 
In the normalized notation, the new parameters are
$w'_i = \frac{\ell_1}{\ell_i}$.
Again, the allocation is unchanged whether we use the original notation or the normalized one:
\[
    x_{i,j}(\hat{G},{\bf \hat{w}}') = x_{i,j}(G,\bw') = \frac{g_{i,j}\cdot w'_i}{
\sum_{i' \in [m]} g_{i',j} \cdot w'_{i'}},
\]
and we have, 
$\ell'_i = \ell_i(P,\hat{G},\hat{\bw'}) = 
\ell_i(P,G,\bw')$.

\smallskip\noindent{\bf The case of Two Agents.}
First, we consider the case of two agents here, i.e., $m=2$. Later, we will show the reduction from general $m$ to $m=2$.

We have
\begin{align*}
    \ell_1 = \sum_j g_{1,j}\cdot p_{1,j}
    \qquad \text{ and } \qquad
    \ell_2 = \sum_j g_{2,j}\cdot p_{2,j},
\end{align*}
and the parameter for the second agent after the update is given by:
$w'_2= \frac{\ell_1}{\ell_2}$ (note that $w'_1=1)$.

Accordingly, the loads after the update are given by:
\begin{align*}
\ell'_1 =  \sum_j p_{1,j} \cdot \frac{g_{1,j}}{g_{1,j}+w'_2 \cdot g_{2,j}}
\qquad \text{ and } \qquad
\ell'_2 =  \sum_j p_{2,j} \cdot \frac{w'_2 \cdot g_{2,j}}{g_{1,j}+w'_2 \cdot g_{2,j}}.
\end{align*}

Assume w.l.o.g that $\ell_1 < \ell_2$.
First, note that, from monotonicity (Observation~\ref{obv:monotone}) we have:
$$
\ell'_2 \leq \ell_2 = \nicefrac{\elx}{\left(1+\frac{\elx-\ell_2}{\tilde{p}_1}\right)} .
$$
Next,  we have to show that 
\begin{equation}\label{eq:upper}
\ell'_1 \leq \nicefrac{\elx}{\left(1+\frac{\elx-\ell_1}{\tilde{p}_1}\right)}
=
\nicefrac{\ell_2}{\left(1+\frac{\ell_2-\ell_1}{\tilde{p}_1}\right)}.
\end{equation}
The proof of the lower bound on $\ell'_1$ is similar and is omitted for brevity. 

We use the following standard inequality:
\begin{fact}[Milne's Inequality~\cite{milne1925note}]\label{fact:milne}
For any $a,b \in \R^n$, we have
$$ \sumdsj \frac{a_j \cdot b_j}{a_j+b_j} \quad \leq \quad \frac{\sumdsj a_j \cdot \sumdsj b_j}{\sumdsj {(a_j+b_j)}}.$$
\end{fact}

In using this inequality, we set for any $j\in [n]$, 
\[
    a_j = p_{1,j} \text{ and } b_j = p_{1,j}\cdot \left(\frac{f_j}{w'_2}-1\right) \text{  where } f_j = g_{1,j}+w'_2 \cdot g_{2,j} = g_{1,j}+w'_2 \cdot (1-g_{1,j}).
\]
First, we calculate each term in Milne's inequality separately:
\begin{align*}
    \sumdsj \frac{a_j \cdot b_j}{a_j+b_j} 
    &= \sumdsj p_{1,j} \cdot \frac{f_j-w'_2}{f_j}
    = \sumdsj p_{1,j} \cdot \frac{g_{1,j}+w'_2\cdot g_{2,j}-w'_2}{f_j}
    = \sumdsj p_{1,j} \cdot \frac{g_{1,j}-w'_2\cdot(1- g_{2,j})}{f_j}\\
    &=\sumdsj p_{1,j} \cdot \frac{g_{1,j}-w'_2 \cdot g_{1,j}}{f_j}
    = \sumdsj p_{1,j} \cdot g_{1,j}\cdot \frac{1-w'_2}{f_j} 
    = \ell'_1\cdot(1-w'_2).\\
    \sumdsj a_j &= \tilde{p}_1.\\
    \sumdsj b_j &= \sumdsj p_{1,j}\cdot g_{1,j}\cdot \left(\frac{1}{w'_2}-1\right)  = \frac{\ell_1}{w'_2}-\ell_1 = \ell_2-\ell_1 =
    \ell_2\cdot (1-w'_2).
\end{align*}

Using \Cref{fact:milne}, we get 

\begin{equation*}
    \label{eq:convergence}
    \ell'_1\cdot(1-w'_2) \leq  \frac{\tilde{p}_1 \cdot \ell_2 }{\ell_2-\ell_1+\tilde{p}_1}  \cdot (1-w'_2)
\end{equation*} 

By our assumption that $\ell_1 <  \ell_2$, and therefore $w'_2<1$. We now get \Cref{eq:upper} by rearranging terms. 
This completes the proof for the lemma for the case of two agents. 


\begin{toappendix}
    
\smallskip\noindent{\bf General case of More than Two Agents.}
For more than two agents, we again only show the upper bound; the lower bound follows similarly. We also focus on agent $1$ which is w.l.o.g. by symmetry. Therefore, we have to show that:
\begin{equation}\label{eq:upper-gen}
    \ell_1' \leq \frac{\elx}{1+\frac{\elx-\ell_1}{\tilde{p}_1}}.
\end{equation}    

To show this inequality, we use a two-step transformation to an instance with two agents. In the first step, we change the weight matrix by increasing the weights of jobs for agents other than agent $1$ so that the loads of all agents except $1$ becomes $\elx$. We argue below that this is w.l.o.g. 
%
In the second step, we transform the instance to two agents, where we ``combine'' all the other $m-1$ agents (except agent $1$) to a single row in the matrices $P$ and $G$ (this represents the second agent in the transformed instance).
Again, we show that we can do this in a way that establishing the upper bounds on $\ell'_1$ after the transformation implies \Cref{eq:upper-gen}. Finally, we use \Cref{eq:upper} to conclude the proof.

\smallskip\noindent{\bf First transformation:}
We assume $G,\bw$ are normalized as earlier. Recall that in this case, we have $\bw = \mathbf{1}^m$.
Consider the instance $\hat{P},\hat{G}$, where $\hat{p}_{i,j} = p_{i,j} \cdot \frac{\elx}{\ell_i}$ and $\hat{p}_{i,1} = p_{i,1}$, and $\hat{G} = G$.
Let the corresponding load in the transformed instance
be denoted $\hat{\bell} = \bell(\hat{P},\hat{G},\mathbf{1}^m)$.
By definition, $\ell_1 = \hat{\ell}_1$, and 
$\elx = \hat{\ell}_{\max}$.

Let, $\hat{w}_i = \ell_1/\hat{\ell}_i$ for all $i
\in [m]$. Note that we have $\hat{w}_1 = w'_1 = 1$ and
$\hat{w}_i = \ell_1/\elx \leq \ell_1/\hat{\ell}_i = w'_i$
for $i\geq 2$. By \Cref{obv:monotone}, 
we have $\ell'_1 = \ell_1(P,G,\bw') \leq \ell_1(P,G,\hat{\bw}') =
\ell_1(\hat{P},\hat{G},\hat{\bw}')$. Thus, it suffices to show
\Cref{eq:upper-gen} on the transformed instance.

\smallskip\noindent{\bf Second transformation:}
Now, define $\tilde{P},\tilde{G} \in \Rp^{2 \times n}$ as follows:
\begin{eqnarray*}
\tilde{g}_{1,j} &=& \hat{g}_{1,j} \quad = \quad g_{1,j}\\
\tilde{g}_{2,j} &=& 1-g_{1,j}\\
\tilde{p}_{1,j} &=& {p}_{1,j}\\
\tilde{p}_{2,j} &=& \frac{\sum_{i=2}^m \hat{p}_{i,j}\cdot g_{i,j}}{(m-1)\cdot (1-g_{1,i})}.\\
\end{eqnarray*}

Before the update, we update
$\ell_1(\hat{P},\hat{G},\mathbf{1}^m) = \ell_1(\tilde{P},\tilde{G},\mathbf{1}^2)$ since we did not modify the rows in $P$ and $G$ corresponding to agent $1$.
Second, we have
$$ \ell_2(\tilde{P},\tilde{G},\mathbf{1}^2) = \sum_j \tilde{p}_{2,j}\cdot \tilde{g}_{2,j} = \frac{\sum_{i=2}^m \hat{p}_{i,j}\cdot g_{i,j}}{(m-1)\cdot (1-g_{1,i})} \cdot \tilde{g}_{2,j}= \frac{(m-1)
\cdot \elx}{(m-1)} = \elx$$
For $\tilde{\bw}$ such that,
$ \tilde{w}'_1 = 1$ and $\tilde{w}'_2 = \hat{\ell}_1/\elx$, we have 
$$x_{1,j}(\tilde{G},\tilde{w}') = \frac{\tilde{g}_{1,j}}{\tilde{g}_{1,j}+\hat{\ell}_1/\elx \cdot \tilde{g}_{2,j}} = \frac{g_{1,j}}{g_{1,j}+\hat{\ell}_1/\elx \cdot \sum_{i=2}^m g_{i,j}}  = x_{1,j}(\hat{G},\hat{w}').$$
Therefore, $\ell_1(\hat{P},\hat{G},\hat{w}') = \ell_1(\tilde{P},\tilde{G},\tilde{w}')$.
Finally, by the case of two agents (\Cref{eq:upper}), we have
\[
    \ell_1(\tilde{P},\tilde{G},\tilde{w}') \leq \frac{\elx}{1+\frac{\elx-\ell_1}{\tilde{p}_1}}
\]    
    and therefore
$$\ell_1(P,G,w') \leq \ell_1(\hat{P},\hat{G},\hat{w}') = \ell_1(\tilde{P},\tilde{G},\tilde{w}') \leq \frac{\elx}{1+\frac{\elx-\ell_1}{\tilde{p}_1}},$$
as required.

\end{toappendix}

\begin{toappendix} 

\subsection{Strict Monotonicity of the Ratio of the Maximum to Minimum Loads in \Cref{alg:iterative}}\label{sec:strict-monotone}

We will need the following observation, which relates the assignment vectors for two different parameter vectors. We will use this later to relate the assignment vectors for an agent before and after a single iteration of \Cref{alg:iterative}.
\begin{observation}
\label{obv:xval}
Fix any $G\in \mps$. Consider two parameter vectors $\bw$ and $\bw'$ where we denote their coordinate-wise ratio as $\tau_i = \frac{w'_i}{w_i}$ for all $i\in [m]$.
Let $y_{i,j} = x_{i,j}(G,\bw)$ and $z_{i,j} = x_{i,j}(G,\bw')$ be the fractional allocations corresponding to the parameter vector $\bw, \bw'$ respectively. Then we have
\[
z_{i,j} = \frac{\tau_i\cdot y_{i,j}}{\sum_{i'\in[m]} \tau_{i'} \cdot y_{i',j}},
\]
and
\[
\frac{y_{i,j}}{z_{i,j}} = 
{\sum_{i'\in[m]} \frac{\tau_{i'}}{\tau_i} \cdot y_{i',j}},
\]

\end{observation}

Next, we show that for fixed matrices $P$ and $G$, the assignment variable $x_{i, j}$ is at least some fixed value.

\begin{lemma}\label{lem:minimum}
Let $P, G\in \mps$ be given fixed matrices. Then, for any iteration $r$ of \Cref{alg:iterative} and the corresponding parameter vector $\bw^{(r)}$, we have $x_{i, j} (G, \bw^{(r)}) \ge \xmin$ for some fixed $\xmin > 0$ that depends only on $P$ and $G$.
\end{lemma}
\begin{proof}
First,  we show that for every iteration $r$ in \Cref{alg:iterative}, and for any two agents $i',i\in [m]$, the ratio of their respective parameters $\frac{w^{(r)}_{i'}}{w^{(r)}_{i}}$ is bounded by a term that only depends on the matrices $P$ and $G$. 
To obtain this bound, we define two terms that depend only on the matrices $P$ and $G$. The first term, denoted $\alpha$, is the ratio of the maximum to minimum load at the beginning of \Cref{alg:iterative}, i.e., $\alpha = \frac{\elx^{(0)}}{\eln^{(0)}}$. The second term, denoted $\rho_{i, i'}$, is specific to the agents $i, i'$ and is defined as 
$\rho_{i,i'} =  \max_{j\in [n]} \biggl \{\frac{p_{i,j}}{p_{i', j}} \cdot \frac{g_{i,j}}{g_{i',j}}\biggl\}$. 
Our goal is to show that for every iteration $r$ of \Cref{alg:iterative}, we have $\frac{w^{(r)}_{i'}}{w^{(r)}_{i}} \le  \alpha\cdot \max(\rho_{i, i'}, 1)$.

We show this bound in two steps. First, we show that the ratio $\frac{w^{(r)}_{i'}}{w^{(r)}_{i}}$ cannot increase by a factor greater than $\alpha$ in any iteration. Next, we show that if this ratio $\frac{w^{(r)}_{i'}}{w^{(r)}_{i}}$ exceeds $\rho_{i, i'}$ in any iteration, then it must decrease in the next iteration. Further, observe that the initial value of this ratio $\frac{w^{(r)}_{i'}}{w^{(r)}_{i}}$ is 1 for every pair of agents $i, i'\in [m]$ since $w^{(0)}_i = 1$ for all agents $i\in [m]$. Putting these together, we can then claim that $\frac{w^{(r)}_{i'}}{w^{(r)}_{i}} \le  \alpha\cdot \max(\rho_{i, i'},1 )$ for all iterations $r$ and for any two agents $i, i'\in [m]$.

We first prove that $\frac{w^{(r)}_{i'}}{w^{(r)}_{i}}$ cannot increase by a factor greater than $\alpha$ in any iteration. We have the following: 
\begin{eqnarray*}
    \frac{w^{(r)}_{i'}}{w^{(r)}_{i}} 
    &= \frac{\ell_{i}^{(r-1)}}{\ell_{i'}^{(r-1)}} \cdot \frac{w^{(r-1)}_{i'}}{w^{(r-1)}_{i}} \quad & \quad \text{(by the definition of \Cref{alg:iterative})}\\
    &\leq \frac{\elx^{(r-1)}}{\eln^{(r-1)}} \cdot \frac{w^{(r-1)}_{i'}}{w^{(r-1)}_{i}} \quad & \quad \text{(by the definition of $\elx$ and $\eln$)}\\
    &\leq \frac{\elx^{(0)}}{\eln^{(0)}} \cdot \frac{w^{(r-1)}_{i'}}{w^{(r-1)}_{i}} \quad & \quad \text{(by \Cref{lem:notincreasing})}\\
    &= \alpha \cdot \frac{w^{(r-1)}_{i'}}{w^{(r-1)}_{i}} \quad & \quad \text{(by the definition of $\alpha$)}.
\end{eqnarray*}

Next, we prove that if $\frac{w^{(r)}_{i'}}{w^{(r)}_{i}} > \rho_{i, i'}$ in any iteration $r$, then the ratio must decrease in the next iteration, i.e., $\frac{w^{(r+1)}_{i'}}{w^{(r+1)}_{i}} < \frac{w^{(r)}_{i'}}{w^{(r)}_{i}}$. Note that if $\frac{w^{(r)}_{i'}}{w^{(r)}_{i}} > \rho_{i, i'}$, this implies that 
\begin{equation}\label{eq:cond}
    \frac{w^{(r)}_{i'}}{w^{(r)}_{i}} > \frac{p_{i,j}}{p_{i', j}} \cdot \frac{g_{i,j}}{g_{i',j}} \quad  \text{ for every item } j\in [n], 
\end{equation}    
since $\rho_{i, i'} = \max_{j\in [n]} \biggl\{ \frac{p_{i,j}}{p_{i', j}} \cdot \frac{g_{i,j}}{g_{i',j}}\biggl\}$. Now, by the rules of proportional allocation, we have for every item $j\in [n]$:
\[
    \frac{x^{(r)}_{i', j}}{x^{(r)}_{i, j}} = \frac{w^{(r)}_{i'}}{w^{(r)}_{i}} \cdot \frac{g_{i', j}}{g_{i, j}} > \left(\frac{p_{i,j}}{p_{i', j}} \cdot \frac{g_{i,j}}{g_{i',j}}\right)\cdot \frac{g_{i', j}}{g_{i, j}} = \frac{p_{i,j}}{p_{i', j}} \quad \text{(the inequality is from \Cref{eq:cond})}.
\]
Then, the loads of the agents $i, i'$ in iteration $r$ of \Cref{alg:iterative} satisfy
\[
   \ell^{(r)}_{i} =  \sum_{j\in [n]} x^{(r)}_{i, j}\cdot p_{i, j} < \sum_{j\in [n]} x^{(r)}_{i', j}\cdot p_{i', j} = \ell^{(r)}_{i'}.
\]
Then,
\[
     \frac{w^{(r+1)}_{i'}}{w^{(r+1)}_{i}} 
     = \frac{\ell_{i}^{(r)}}{\ell_{i'}^{(r)}} \cdot \frac{w^{(r)}_{i'}}{w^{(r)}_{i}} 
     < \frac{w^{(r)}_{i'}}{w^{(r)}_{i}}.
\]

We have now shown $\frac{w^{(r)}_{i'}}{w^{(r)}_{i}} \le  \alpha\cdot \max(\rho_{i, i'},1 )$ for all iterations $r$ and for any two agents $i, i'\in [m]$. In other words, $\frac{w^{(r)}_{i'}}{w^{(r)}_{i}} \le  \tau$, where we define $\tau := \alpha \cdot \max \biggl(\max_{i, i'} \rho_{i, i'},1\bigg)$.

Now, recall that 
\[
x^{(r)}_{i, j} 
= \frac{w^{(r)}_i\cdot g_{i, j}}{\sum_{i'}w^{(r)}_{i'} \cdot g_{i', j}}
= \frac{g_{i, j}}{g_{i, j} + \sum_{i'\not= i} \left(\frac{w^{(r)}_{i'}}{w^{(r)}_i}\right) \cdot g_{i', j}}
\ge \frac{g_{i, j}}{g_{i, j} + \sum_{i'\not= i} \tau \cdot g_{i', j}}, \text{ since } \frac{w^{(r)}_{i'}}{w^{(r)}_{i}} \le  \tau.
\]
To complete the proof, we define $\xmin = \min_{i, j} \frac{g_{i, j}}{g_{i, j} + \sum_{i'\not= i} \tau \cdot g_{i', j}}$. Note that $\xmin$ only depends on $P$ and $G$ as required by the lemma.
\end{proof}

We are now ready to show the strict monotonicity property. Note that since $\elx^{(r+1)} \le \elx^{(r)}$ by the weak monotonicity property (\Cref{lem:notincreasing}), it suffices to show that $\eln^{(r+1)} - \eln^{(r)}$ is sufficiently large so that the ratio $\nicefrac{\elx}{\eln}$ converges to 1. We bound the increase in $\eln$ in the next lemma, and then show the convergence in the proof of \Cref{thm:maincanonical}.

\begin{proof}[Proof of \Cref{lem:strictsingle}]
We will prove that the minimum load $\eln$ will strictly increase in the next iteration; specifically that 
$\eln^{(r+1)} \geq (1+c \cdot \epsilon)\cdot \eln^{(r)}$, for some constant $c > 0$ that only depends on $P$ and $G$. 

Let $\delta = \xmin\cdot \big(1-\frac{1}{1+{\epsilon}}\big)$, where $x_{\min}$ is as defined in \Cref{lem:minimum}. We divide the agents into two sets: 
the {\em light} agents $S_s = \left\{i \in [m]: \ell^{(r)}_i \leq (1+\delta) \cdot \eln^{(r)}\right\}$
and the {\em heavy} agents $S_t = \left\{i \in [m]: \ell^{(r)}_i > (1+\delta) \cdot \eln^{(r)}\right\}$.
The bulk of our proof bounds the increase in the load of every light agent $i\in S_s$. For every heavy agent $i\in S_t$, we use \Cref{lem:notincreasing} to show that its load in iteration $r+1$ is sufficiently large. Putting these together yields the lemma.

First, let us consider a light agent $i\in S_s$. For each item $j\in [n]$, define $y_{i,j} = x_{i,j}(G,\bw^{(r)})$ and $z_{i,j} = x_{i,j}(G,\bw^{(r+1)})$. 
First, we show that for each item $j\in [n]$,
\begin{equation}\label{eq:diff}
    z_{i,j}-y_{i,j}\geq x_{\min}^2 \cdot \delta.
\end{equation}


For any agent $i\in [m]$, let $\tau_i = \frac{w_i^{(r+1)}}{w_i^{(r)}}$. Then, we have:
\begin{eqnarray*}
\frac{y_{i,j}}{z_{i,j}} 
 & =  & \frac{\sum_{i'\in[m]} \tau_{i'} \cdot y_{i',j}}{\tau_i} \qquad \text{(By \Cref{obv:xval})} \\  
 & = & \sum_{i'\in[m]} \frac{w_{i'}^{(r+1)}}{w_{i'}^{(r)}}\cdot \frac{w_{i}^{(r)}}{w_{i}^{(r+1)}} \cdot y_{i',j} =  \displaystyle\sum_{i'\in[m]} \frac{\ell^{(r)}_{i}}{\ell^{(r)}_{i'}} \cdot y_{i',j} \qquad \text{(By the definition of \Cref{alg:iterative})}.
\end{eqnarray*}
Now, let $k$ be an agent with maximum load in iteration $r$, i.e, $k \in \arg \max_i \ell^{(r)}_i$. We rewrite the above equation as:
\begin{eqnarray*}
     \frac{y_{i,j}}{z_{i,j}}  
     & = & y_{i,j}+ y_{k,j}\cdot\frac{\ell_i^{(r)}}{\elx^{(r)}} + \sum_{i'\in [m] \setminus \{i,k\}} \frac{\ell^{(r)}_{i}}{\ell^{(r)}_{i'}} \cdot y_{i',j}\\
     & \le & y_{i,j}+ y_{k,j}\cdot\frac{\ell_i^{(r)}}{\elx^{(r)}} +  \frac{\ell^{(r)}_{i}}{\eln^{(r)}} \cdot \sum_{i'\in [m] \setminus \{i,k\}} y_{i',j} \quad \text{(since $\ell^{(r)}_{i'} \ge \eln^{(r)}$ for all $i'\in [m]$)}\\
     & \le & 1\cdot y_{i,j}+ \frac{\ell_i^{(r)}}{\elx^{(r)}}\cdot y_{k,j} + \frac{\ell^{(r)}_{i}}{\eln^{(r)}}\cdot \left(\sum_{i'\in [m] \setminus \{i,k\}} y_{i',j}\right).
\end{eqnarray*}
Now, note that $\sum_i y_{i, j} = y_{i, j} + y_{k, j} + \sum_{i'\in [m] \setminus \{i,k\}} y_{i',j} = 1$, i.e., the RHS of the above inequality is a convex combination of $1, \frac{\ell_i^{(r)}}{\elx^{(r)}}$, and $\frac{\ell^{(r)}_{i}}{\eln^{(r)}}$. Now, since $\frac{\ell_i^{(r)}}{\elx^{(r)}} \le 1 \le  \frac{\ell^{(r)}_{i}}{\eln^{(r)}}$, this expression is maximized when $y_{i, j}$ and $y_{k, j}$ are minimized. By \Cref{lem:minimum}, we know $y_{i, j}, y_{k, j} \ge x_{\min}$. Hence, we can write
\begin{align*}
\frac{y_{i,j}}{z_{i,j}} & \leq  \xmin+ \xmin\cdot\frac{\ell_i^{(r)}}{\elx^{(r)}} +  (1-2\cdot \xmin)\cdot \frac{\ell_i^{(r)}}{\eln^{(r)}} 
\\ & \leq \xmin+ \xmin\cdot\frac{1+\delta}{1+\epsilon} +  (1-2\cdot \xmin)\cdot (1+\delta)
\\& \qquad \text{(since $i\in S_s$, $\ell^{(r)}_i\le (1+\delta)\cdot \eln^{(r)}$, and by definition of $\epsilon$, $\elx^{(r)} 
\geq (1+\epsilon)\cdot \eln^{(r)}$)} 
\\ & =   \xmin+ \frac{\xmin}{1+\epsilon} + \frac{\xmin\cdot\delta}{1+\epsilon} + 1+\delta -2\cdot \xmin -2\cdot \xmin\cdot \delta 
\\ & \leq   \xmin+ \frac{\xmin}{1+\epsilon} + \xmin\cdot\delta + 1+\delta -2\cdot \xmin -2\cdot \xmin\cdot \delta \qquad \qquad \text{(Since $\epsilon>0$)}
\\ & = 1 - \xmin\cdot \delta + \delta - \xmin \cdot (1-\nicefrac{1}{1+\epsilon}) 
\\ & = 1 - \xmin\cdot \delta. \qquad \qquad \text{(By the definition of $\delta$ )}
\end{align*}

Therefore, for any $i\in S_s$ and for any  $j\in[n]$, 
\[
    \frac{z_{i,j}}{y_{i,j}} \geq \frac{1}{1 - \xmin\cdot \delta } \geq 1+\xmin\cdot \delta.
\] 

Note that $y_{i,j}\geq \xmin$ by \Cref{lem:minimum}. Hence,
$z_{i,j}- y_{i,j} \geq \xmin^2\cdot \delta$. This establishes \Cref{eq:diff}.

Now, recall that $\tilde{p}_i = \sum_j p_{i, j}$ for all $i\in [m]$. Now, let $\tilde{p}_{\min} = \displaystyle\min_{i\in[n]} \tilde{p}_i$. 
We have 
\[
\ell_i^{(r+1)} 
= \ell_i^{(r)} + \sum_{j\in [n]}(z_{i, j}-y_{i, j})\cdot p_{i, j} 
\geq \eln^{(r)} + \delta\cdot \xmin^2\cdot \tilde{p}_{\min} \quad \text{(by \Cref{eq:diff})}.
\]

Now, let $c_3 = \xmin^2 \cdot \frac{\tilde{p}_{\min}}{\elx^{(0)}}$.
By \Cref{lem:notincreasing}, we have $c_3 \le \xmin^2 \cdot \frac{\tilde{p}_{\min}}{\elx^{(r)}}$, and therefore, $c_3 \leq \xmin^2 \cdot \frac{\tilde{p}_{\min}}{\eln^{(r)}}$ since $\eln^{(r)} \le \elx^{(r)}$. Therefore, we can write the above inequality as:
\[ 
\ell_i^{(r+1)} 
\geq \eln^{(r)} + \delta\cdot \xmin^2\cdot \tilde{p}_{\min}
\ge (1+c_3\cdot \delta)\cdot \eln^{(r)} \quad \text{for all light agents } i\in S_s.
\]
Note that $c_3$ depends only on $P$ and $G$.

Finally, we consider heavy agents. Let $c_4 = \eln^{(0)} \cdot \min_{i\in [m]} \frac{1}{\tilde{p}_i} \leq \eln^{(r)} \cdot \min_{i\in [m]} \frac{1}{\tilde{p}_i}$ by \Cref{lem:notincreasing}.
For all $i\in S_t$, we have 
\begin{eqnarray*}
    \ell_i^{(r+1)} 
    & \geq & \frac{\eln^{(r)}}{\left(1-\frac{\ell^{(r)}_i-\eln^{(r)}}{\tilde{p}_i}\right)} \qquad \text{(by \Cref{lem:notincreasing})}\\
    & \geq & \eln^{(r)}\cdot {\left(1+\frac{\ell^{(r)}_i-\eln^{(r)}}{\tilde{p}_i}\right)} \qquad \text{}\\
    & \geq & \eln^{(r)}\cdot \left(1+\eln^{(r)} \cdot \frac{\delta}{\tilde{p}_i}\right) \qquad \text{(since $i\in S_t$)}\\
    & \geq & \eln^{(r)}\cdot (1+c_4 \cdot \delta).
\end{eqnarray*}    

Thus, we have established that for all agents $i\in [m]$, we have 
\[
    \ell_i^{(r+1)} \ge (1+ \min(c_3, c_4)\cdot \delta))\cdot \eln^{(r)}.
\]
Now, $\delta = \frac{\epsilon}{1+\epsilon}\cdot x_{\min} \geq (\nicefrac{\epsilon}{2})\cdot \xmin$ since $\epsilon \in (0, 1]$.
Let us define $c = \min (c_3, c_4)\cdot (\nicefrac{x_{\min}}{2})$. Therefore, we get that for all agents $i\in [m]$, it holds that 
\[
    \ell_i^{(r+1)} \ge (1+ c\cdot \epsilon)\cdot \eln^{(r)}, \text{ as desired}.
\]
\end{proof}

\end{toappendix}


\section{Monotonicity and Convergence of Exponentiated Proportional Allocations}
\label{sec:ep-monotonicity-convergence}

In this section, we prove the monotonicity and convergence of \ep-allocations (\Cref{thm:exponential}).

First, we establish monotonicity of \ep-allocations (first part of \Cref{thm:exponential}).
We compare two \ep-allocations with arbitrary learned parameters but different exponential constants. We show that with a larger exponent, at least one agent's load will be higher, regardless of the parameters used.

\begin{lemma}
\label{lem:monotoneaid}
Fix a weight matrix $P\in \mps$. Let $\alpha,{{\alpha'}} \in \R$ such that $\alpha > {\alpha'}$. Now, for any two sets of learned parameters $\bw_\alpha,\bw_{\alpha'}\in \vps$, there exists an agent $k\in [m]$ such that 
\[
    \ell_k(P,\alpha,\bw_\alpha) \geq \ell_k(P,{\alpha'},\bw_{\alpha'}).
\]    
\end{lemma}

\begin{proof}

Let $\Ldiff$ denote the vector of differences of loads of the machines in the two allocations, 
namely $\Ldiff_i = \ell_i(P,\alpha,\bw_\alpha) - \ell_i(P,{\alpha'},\bw_{\alpha'})$.
Our goal is to show that $\Ldiff$ has at least one nonnegative coordinate.

To show this, we define a vector in the positive orthant 
$\mathbf{c} \in \vps$ as follows:
\[
    c_i = \left(\frac{w_{\alpha,i}}{w_{{\alpha'},i}}\right)^{\frac{1}{\rho}}, 
    \text{ where } \rho=\alpha-{\alpha'} > 0
\]    
and show that this vector $\bf c$ has a nonnegative inner product with the vector 
$\Ldiff$. Note that this suffices since the inner product of a vector with all
positive coordinates and one with all negative coordinates cannot be nonnegative.
In other words, we want to show the following:
\begin{equation}\label{eq:overall}
    \sum_{i\in[m]} c_i \cdot (\ell_i(P,\alpha,w_\alpha) - \ell_i(P,{\alpha'},w_{\alpha'})) \geq 0.
\end{equation}  

Let us denote the fractional allocation of an item $j$ in the two cases by $x_{i, j}$
and $x'_{i, j}$ respectively. Then, \Cref{eq:overall} can be rewritten as
\[
    \sum_{i\in [m]} c_i \cdot \sum_{j\in [n]} p_{i, j}\cdot (x_{i, j} - x'_{i, j})\geq 0.
\]
Changing the order of the two summations, we rewrite further as
\[
    \sum_{j\in [n]} \left(\sum_{i \in [m]}  c_i\cdot p_{i, j}\cdot (x_{i, j} - x'_{i, j})\right)\geq 0.
\]
We will prove this inequality separately for each item $j\in [n]$. Namely, we will show that 
\begin{equation}\label{eq:perjob}
    \sum_{i\in [m]}  c_i\cdot p_{i, j}\cdot (x_{i, j} - x'_{i, j}) \ge 0, \text{ for every } j\in [n].
\end{equation}

Fix an item $j$. Since the item is fixed, we will drop $j$ from the notation and define 
$\bu \in \R^m$ as
\[
u_{i} = p_{i}\cdot (x_{i}-x'_{i}).
\]
So, we need to show that 
\begin{equation}\label{eq:monotonicity}
    {\bf c}\cdot \bu \ge 0, \text{ i.e., } \sum_{i\in [m]} c_i \cdot u_{i} \geq 0.
\end{equation}    

We have
\begin{eqnarray*}
\sum_i c_i \cdot u_i 
& = & \sum_i {c_i \cdot p_i \cdot \left(\frac{p_i^\alpha\cdot w_{\alpha,i}}{\sum_{i'}p_{i'}^\alpha\cdot w_{\alpha,i'}}
    -\frac{p_i^{\alpha'}\cdot w_{{\alpha'},i}}{\sum_{i'}p_{i'}^{\alpha'}\cdot w_{{\alpha'},i'}}\right)}\\
&=& \frac{1}{T}\cdot \sum_i {c_i \cdot p_i \cdot \left({
p_i^\alpha\cdot w_{\alpha,i}\cdot \left( \sum_{i'}p_{i'}^{\alpha'}\cdot w_{{\alpha'},i'}\right) -
p_i^{\alpha'}\cdot w_{{\alpha'},i}\cdot \left( \sum_{i'}p_{i'}^\alpha\cdot w_{\alpha,i'}\right)
}\right)} \\
& & \qquad \text{ where } T=\left(\sum_{i'}p_{i'}^{\alpha'}\cdot w_{{\alpha'},i'}\right)\cdot\left(\sum_{i'}p_{i'}^\alpha\cdot w_{\alpha,i'}\right).\\
\end{eqnarray*}
Now, on the right hand side of the above equation, we replace $\alpha$ by $\alpha'+\rho$ and $w_{\alpha, i}$ by $w_{\alpha', i}\cdot c_i^\rho$ for every $i\in [m]$. This gives us:
\begin{flalign*}
&\sum_i c_i \cdot u_i 
= \\
&\frac{1}{T} \sum_i {c_i \cdot p_i  \left({p_i^{\alpha'}\cdot p_i^\rho\cdot w_{{\alpha'},i}\cdot c_{i}^\rho \left( \sum_{i'}p_{i'}^{\alpha'}\cdot w_{{\alpha'},i'}\right) - p_i^{\alpha'}\cdot w_{{\alpha'},i} \left( \sum_{i'}p_{i'}^{\alpha'}\cdot p_{i'}^\rho\cdot w_{{\alpha'},i'}\cdot c_{i'}^\rho\right)}\right)} \\
&= \frac{1}{T}  \sum_i {b_i  \left({a_i\cdot b_{i}^\rho \left( \sum_{i'}a_{i'}\right) - a_i \left( \sum_{i'}a_{i'}\cdot b_{i'}^\rho\right)}\right)}, \\
& \qquad \text{ where } a_i = w_{{\alpha'},i}\cdot p_i^{\alpha'} \text{ and } b_i = p_i \cdot c_i.
\end{flalign*}
Rearranging the summations on the two terms on the right hand side, we get
\begin{equation*}
\sum_i c_i \cdot u_i 
= \frac{1}{T}\cdot  \left( \sum_{i'}a_{i'}\right) \cdot \sum_i { a_i\cdot b_{i}^{\rho+1}} - \frac{1}{T}\cdot \left( \sum_{i'}a_{i'} \cdot b_{i'}^\rho\right) \cdot \sum_i { a_i\cdot b_{i}}
\end{equation*}
Now, let $z_i = a_i^{1/2}$, and $y_i=a_i^{1/2} \cdot b_i^{\rho/2+1/2}$, and $\theta=\frac{|\rho-1|}{\rho+1}$. Then, we have
\begin{eqnarray*}
T\cdot \sum_i c_i \cdot u_i 
& = & \left(\sum_{i'}a_{i'}\right)\cdot \left(\sum_i { a_i\cdot b_{i}^{\rho+1}}\right) - \left(\sum_{i'}a_{i'} \cdot b_{i'}^\rho \right)\cdot \left(\sum_i { a_i\cdot b_{i}}\right) \\
&=&  \left(\sum_{i'}z_{i'}^2\right) \cdot \left(\sum_i { y_i^2}\right) - \left(\sum_{i'}z_{i'}^{1+\theta}\cdot y_{i'}^{1-\theta}\right) \cdot \left(\sum_i { z_i^{1-\theta}\cdot y_i^{1+\theta}}\right).
\end{eqnarray*}
In the last equation, the first term follows directly from $a_{i'} = z_{i'}^2$ and $a_i\cdot b_{i}^{\rho+1} = y_i^2$. The second term is more complicated. There are two cases. If $\rho \le 1$, then $a_{i'} \cdot b_{i'}^\rho = z_{i'}^{1+\theta}\cdot y_{i'}^{1-\theta}$ and $a_i\cdot b_{i} = z_i^{1-\theta}\cdot y_i^{1+\theta}$ but if $\rho > 1$, then the roles get reversed and we get $a_{i'} \cdot b_{i'}^\rho = z_{i'}^{1-\theta}\cdot y_{i'}^{1+\theta}$ and $a_i\cdot b_{i} = z_{i}^{1+\theta}\cdot y_{i}^{1-\theta}$.

Now, note that $T\ge 0$. So, to establish $\sum_i c_i \cdot u_i \ge 0$, it suffices to show that the right hard side of the equation is nonnegative. We do so by employing  Callebaut's inequality which we state below:
\begin{fact}[Callebaut's Inequality~\cite{callebaut1965generalization}] \label{fact:callebaut}
For any $y,z \in \R^n$ and $\theta \le 1$, we have
$$\left(\sum_{i'}z_{i'}^2\right) \cdot \left(\sum_i { y_i^2}\right) \geq \left(\sum_{i'}z_{i'}^{1+\theta}\cdot y_{i'}^{1-\theta}\right) \cdot \left(\sum_i { z_i^{1-\theta}\cdot y_i^{1+\theta}}\right) $$
\end{fact}
Note that we can apply Callebaut's inequality because $\rho\geq0$ implies that $\theta\leq 1$. This completes the proof of the lemma.

\end{proof}


%


\begin{lemma}\label{lem:convergence}
    Given any weight matrix $P \in \mps$ and any constant $\epsilon > 0$, 
    \begin{enumerate}
        \item[(a)] there exists an $\alpha$ (think of $\alpha$ as a sufficiently large negative number) and a corresponding set of parameters $\bw_\alpha$ such that $\ell_i(P,\alpha,w_\alpha) \leq (1+\epsilon) \cdot \lmks(P)$ for all $i\in[m]$.
        \item[(b)] there exists an $\alpha'$ (think of $\alpha'$ as a sufficiently large positive number) and a corresponding set of parameters $\bw_{\alpha'}$ such that $\ell_i(P,\alpha',w_{\alpha'}) \geq (1-\epsilon) \cdot \lsan(P)$ for all $i\in[m]$.
    \end{enumerate}
\end{lemma}

Using \Cref{lem:convergence}, we complete the proof of \Cref{thm:exponential}.
\begin{proof}[Proof of Theorem~\ref{thm:exponential}]
First by Lemma~\ref{lem:unique}, there exists $\bw^*_\alpha$ and $\bw^*_{\alpha'}$,
such that, for all $i \in [m]$, $ \ell_{i}(P,\alpha,\bw^*_\alpha) = \ell^*(P,\alpha)$ and $ \ell_{i}(P,{\alpha'},\bw^*_{\alpha'}) = \ell^*(P,{\alpha'})$. Now, if $\ell^*(P,\alpha) < \ell^*(P,{\alpha'})$, it would contradict Lemma~\ref{lem:monotoneaid}.
And combining \Cref{lem:monotoneaid} and \Cref{lem:convergence}, we completed the proof the second part of \Cref{thm:exponential}.

\end{proof}

\begin{toappendix}
\subsection{Proof of \Cref{lem:convergence}}

We will only prove property (a) in \Cref{lem:convergence} for the \minmax problem. Property (b) for the \maxmin problem has a symmetric proof which is omitted for brevity.

\smallskip\noindent{\bf Properties of an optimal Solution for the \minmax problem.}
First, we establish some properties of an optimal solution. First, we prove the following simple property:
\begin{lemma}
\label{lem:opt-equal}
    Given a weight matrix $P$, the load of every agent $i\in [m]$ in any optimal allocation must be exactly equal to the objective value $\lmks(P)$.
\end{lemma}
\begin{proof}
    If not and there is some agent $k$ with a strictly lower load, then we can remove an infinitesimally small allocation of items from every other agent and assign it to agent $k$ to reduce the objective of the overall allocation.
\end{proof}

Now, given a weight matrix $P\in \vps$ and an optimal solution $x^*$ for the \minmax objective, we define an auxiliary directed graph $G_{x^*}(V,E)$ as follows:
\begin{itemize}
    \item The set of vertices $V = [m]\cup \{0\}$, i.e., the agents and a special vertex labeled $0$.
    \item The set of edges $E = ([m]\times [m]) \cup (\{0\} \times [m])$, i.e., all edges between (ordered) pairs of vertices representing the agents (including self loops) and edges from the special vertex to all the vertices representing the agents. Note that the set of vertices and edges does not depend on $x^*$.
    \item We now define a cost function on the edges that does depend on $x^*$. Edges in $[m]\times [m]$ have the following costs:
    \[ 
        c_{i,k} = \ln \left(\min_{j\in [n]}\left\{\frac{p_{k,j}}{p_{i,j}} \;\bigg\vert\; x^*_{i,j}>0\right\}\right).
    \]
    In other words, the cost of an edge $(i, k)$ is the logarithm of the minimum ratio of the weight of an item for $k$ to that for $i$ among those items that have a non-zero allocation to agent $i$ in $x^*$.
    In addition, all edges incident on the special vertex have cost $0$, i.e., $c_{0,k} = 0$ for all $k\in[m]$.
\end{itemize}

Next, we will show that this graph $G_{x^*}$ does not contain a negative cycle, i.e. a cycle whose edge costs sum to a negative value.

\begin{lemma}
\label{lem:negweight}
    Given a processing matrix $P$ and an optimal solution $x^*$ resulting in an objective value of $\lmks(P)$ for the \minmax problem, the auxiliary graph $G_{x^*}$ does not contain a negative cycle.
\end{lemma}
\begin{proof}
Suppose not, and let $i_1, \dots,i_k, i_{k+1}(=i_1)$ be a negative cycle in $G_{x^*}$. Now, let $j_1,\dots j_k$ be the items that determine the edges costs in this cycle, i.e. 
$j_r = \arg\min_{j\in [n]} \left\{ \frac{p_{i_{r+1},j}}{p_{i_r,j}} \;\big\vert\; x^*_{i_r,j}>0\right\}$.
We have:
\begin{equation}
    \label{eq:negcyc}
\sum_{r=1}^k c_{i_r,i_{r+1}}   = \sum_{r=1}^k  \ln \left(\frac{p_{i_{r+1},j_r}}{p_{i_{r},j_r}}\right) =
\ln \left(\prod_{r=1}^{k} \frac{p_{i_{r+1},j_r}}{p_{i_{r},j_r}}\right) <  0
\end{equation}
Now, define an alternate allocation $x'$ where
$x'_{i_r,j_r} = x^*_{i_r,j_r} - \epsilon_r$,
$x'_{i_{r+1},j_{r}} = x^*_{i_{r+1},j_r} + \epsilon_r$, 
and $x'_{i,j}=x^*_{i,j}$ for all other $i,j$ pairs.
Set $\epsilon_{r+1} = \epsilon_{r} \cdot \frac{p_{i_{r+1},j_{r}}}{p_{i_{r+1},j_{r+1}}}$ with $\epsilon_1 > 0$ being chosen small enough such that $x'$ is a feasible solution (i.e., none of the allocations is negative in $x'$). Note that $x^*_{i_r,j_r} > 0$ and therefore $x^*_{i_{r+1},j_r} < 1$ which implies that such an $\epsilon_1$ exists.

Now, for $r \in [k-1]$ we have $\ell'_{i_{r+1}} - \ell^*_{i_{r+1}} = \epsilon_r \cdot p_{i_{r+1}, j_{r}} - \epsilon_{r+1} \cdot p_{i_{r+1}, j_{r+1}} = 0$. This leaves us to compare the load of $i_1$ in the two allocations. We have
    
$$\ell'_{i_1} - \ell^*_{i_1} 
= \epsilon_k\cdot p_{i_1,j_k} -\epsilon_1 \cdot p_{i_1,j_1}  
= \epsilon_1 \cdot p_{i_1,j_1}\cdot \left(\frac{p_{i_1,j_k}}{p_{i_1,j_1}}\cdot\prod_{r=1}^{k-1} \frac{p_{j_{r},i_{r+1}}}{p_{j_{r+1},i_{r+1}}} -1\right) 
= \epsilon_1 \cdot  p_{i_1,j_1} \cdot \left(\prod_{r=1}^{k} \frac{p_{i_{r+1,j_r}}}{p_{i_r,j_r}} -1\right) 
< 0,
$$ where the last inequality is by (\ref{eq:negcyc}).

This means that the load of an agent can be decreased while keeping all other agents at the same load. But, by \Cref{lem:opt-equal}, all agents must have equal load before this modification since we started with an optimal allocation. This means that after the modification, there is an optimal solution (note that the maximum load has not increased) where the load of agent $i_1$ is lower than the optimal objective $\lmks(P)$. This contradicts \Cref{lem:opt-equal} and therefore completes the proof of this lemma.
\end{proof}

For any agent $i\in [m]$, \Cref{lem:negweight} allows us to define $c^*_i$ as the minimum cost of a path from vertex $0$ to vertex $i$ in the auxiliary graph $G_{x^*}$.
We now define a {\em ratio vector} $\bu \in \vps$, where $u_{i} = e^{c^*_{i}}$.
Our next goal is to show that in the solution $x^*$, the set of agents $S_j$ that an item $j\in [n]$ is allocated to (i.e., $x^*_{i, j} > 0$ if and only if $i\in S_j$) is given by $S_j = \argmin_{i\in [m]}\left(\frac{p_{i,j}}{u_i}\right)$.

\begin{lemma}
\label{lem:restrcitedtransformation}
Given a weight matrix $P$ and an optimal solution $x^*$ resulting in $\lmks(P)$ load, we have that any $i\in [m]$, $j\in [n]$, if there exists some agent $k\in [m]$ such that $\frac{p_{k,j}}{u_k} < \frac{p_{i,j}}{u_i}$, then $x^*_{i,j} = 0$.
\end{lemma}
\begin{proof}
Suppose not. Then, by definition the cost of edge $(i, k)$ in the auxiliary graph $G_{x^*}$ satisfies 
\[
    c_{i,k} \leq \ln\left(\frac{p_{k,j}}{p_{i,j}}\right) < \ln(u_k)-\ln(u_i)
\]    
while 
\[ 
    c^*_{k} \leq c^*_{i} + c_{i,k} < \ln(u_i)+\ln(u_k)-\ln(u_i) = \ln(u_k) = c^*_{k},
\]    
which is a contradiction.
\end{proof}
The following is an immediate corollary:
\begin{corollary}
\label{col:restrcitedtransformation}
For any item $j \in [n]$ and agent $i\in[m]$ such that $x^*_{i,j}>0$, it must be that $i\in S_j$ where $S_j = \arg\min \left\{i \in [m]: \frac{p_{i,j}}{u_i}\right\}$.
\end{corollary}

\smallskip\noindent{\bf Transforming to a restricted related instance.}
We now introduce a special category of instances of the allocation problem that we call {\em restricted related} instances. Such an instance is characterized by a {\em weight vector} $\bp  \in \vpsn$ defined on the items, a {\em speed vector} $\mathbf{v} \in \vps$ defined on the agents, and a binary matrix $E \in \{0,1\}^{m \times n}$ called the {\em admissibility matrix}. Then, the weight $p_{i, j}$ of item $j\in [n]$ for agent $i\in [m]$ is given by the following:
\[
    p_{i, j} = \begin{cases}
                    \frac{p_j}{v_i} & \text{ if } E_{i, j} = 1\\
                    \infty & \text{ if } E_{i, j} = 0.
                \end{cases}
\]

\Cref{col:restrcitedtransformation} allows us to convert any (general) weight matrix $P\in \mps$ to a restricted related instance while preserving the value of $\lmks(P)$. Let $P'\in \mps$ be the weight matrix for the restricted related instance that we construct. Then, we require that the weights for agent-item pairs $i\in [n], j\in [m]$ such that $x^*_{i, j} > 0$ are preserved, while those for the remaining agent-item pairs are made infinite. Clearly, the optimal assignment $x^*$ continues to have the same objective value $\lmks$ after this transformation.

To see why these weights form a restricted related instance, we define the following:
\begin{itemize}
    \item a weight vector $\hat{\bp}  \in \vpsn$ where $\hat{p}_j = \frac{p_{i, j}}{u_i}$ for any $i\in S_j$. (Note that by \Cref{col:restrcitedtransformation}, we get the same value of $\hat{p}_j$ no matter which agent $i\in S_j$ is chosen.)
    \item a speed vector $\hat{\mathbf{v}} \in \vps$ where $\hat{v}_i = \frac{1}{u_i}$. (Note that this implies that the weight of an item $j\in [n]$ for an agent $i\in S_j$ remains unchanged at $p_{i, j}$.)
    \item an admissibility matrix $\hat{E} \in \{0,1\}^{m \times n}$ where $\hat{E}_{i, j} = 1$ if $i\in S_j$ and $0$ if $i\notin S_j$.
\end{itemize}

By Corollary~\ref{col:restrcitedtransformation}, the solution $x^*$ is also feasible for the restricted related instance $(\hat{\bp},\hat{\mathbf{v}},\hat{E})$ and has the same load for every agent $i\in [m]$ since
\[
    \ell_i(P,x^*) = \sum_j x^*_{i,j}\cdot  p_{i,j} = \sum_{j:i\in S_j} x^*_{i,j} \cdot p_{i,j} =  \sum_{j:S_j\ni i} x^*_{i,j}\cdot (\hat{p}_j\cdot {u_i}) = \ell_i(P',x^*),
\]
where $P'$ is the weight matrix of the corresponding restricted related instance.

We now invoke Sinkhorn's theorem (see e.g.~\cite{rothblum1989scalings}) which states the following:
\begin{theorem}[Sinkhorn's Theorem]\label{thm:sinkhorn}
    For any matrix $Z\in \R^{m\times n}$ and vectors ${\bf c}\in \R^n$ and ${\bf r}\in \R^m$, if there is some matrix $Y$ with the properties that (a) the column and row sums of~~$Y$ are equal to $\bf c$ and $\bf r$ respectively and (b) $Y_{i, j} > 0$ only if $Z_{i, j} > 0$, then there exist diagonal matrices $A\in \R^{m\times m}$ and $B\in \R^{n\times n}$ such that the column and row sums of $A\cdot Z\cdot B$ are $\bf c$ and $\bf r$ respectively.
\end{theorem}

To apply \Cref{thm:sinkhorn}, we set $Z$ to the admissibility matrix $\hat{E}$ and the vectors $\bf c$ and $\bf r$ respectively to the vectors $\lmks\cdot \hat{{\bf v}}$ and $\hat{{\bf p}}$. Now, the matrix $Y$ required in the condition for \Cref{thm:sinkhorn} is given by $Y_{i, j} = \hat{p}_j\cdot x^*_{i, j}$. Note that $Y_{i, j} > 0$ only if $i\in S_j$, which in turn implies $\hat{E}_{i, j} = Z_{i, j} = 1$. We therefore apply \Cref{thm:sinkhorn} and obtain diagonal matrices $A\in \R^{m\times m}$ and $B\in \R^{n\times n}$. Finally, we set the vector of parameters $\bw$ to $w_i = A_{i, i}$ to derive the following corollary:

\begin{corollary}
\label{cor:sinkhorn}
There exists a vector of parameters $\bw \in \vps$ such that the following allocation
\[
    \hat{x}_{i,j} = \begin{cases}
                    \frac{w_i}{\sum_{i'\in S_j} w_{i'}} & \text{ if } i\in S_j \\
                    0 & \text{ otherwise }
                    \end{cases}
\]                    
achieves the optimal \minmax objective $\lmks$.
\end{corollary}

We are now ready to finish the proof of \Cref{lem:convergence}.

\begin{proof}[Proof of \Cref{lem:convergence}]

Suppose we are given a weight matrix $P\in \mps$.
By \Cref{col:restrcitedtransformation}  and  \Cref{cor:sinkhorn}, there exists a vector of parameters $\bw$ for the corresponding restricted related instance defined by the ratio vector $\bu$ with the following property: the proportional assignment with these parameters produces an optimal solution $\hat{x}$. Now, for a fixed $\alpha$, let us define
$\hat{w}_{\alpha,i} = \frac{w_i}{u_i^\alpha}$. Then, it is sufficient to show that
$$\lim_{\alpha \rightarrow -\infty} \frac{\hat{w}_{\alpha,i} \cdot p_{i,j}^\alpha}{\sum_{i'}\hat{w}_{\alpha,i'} \cdot p_{i',j}^\alpha} = \hat{x}_{i,j}.$$
We know
$$ \frac{\hat{w}_{\alpha,i} \cdot p_{i,j}^\alpha}{\hat{w}_{\alpha,k} \cdot p_{k,j}^\alpha} = 
\frac{w_i \cdot  p_{i,j}^\alpha / {u_i}^\alpha}{ w_k \cdot p_{k,j}^\alpha/{u_k}^\alpha } = \frac{w_i}{w_k} \cdot \left( \frac{{u_k}\cdot p_{i,j}}{ {u_i} \cdot p_{k,j}} \right)^\alpha.$$
Now, fix an item $j$ and an agent $k \in S_j$. We have the following two cases for any agent $i\in [m]$:
\begin{eqnarray*}
    \frac{\hat{w}_{\alpha, i} \cdot p_{i,j}^\alpha}{\hat{w}_{\alpha,k}\cdot p_{k,j}^\alpha} &=& \frac{w_i}{w_k} \quad  \text{ if } i\in S_j, \text{ and }\\
    \displaystyle\lim_{\alpha \rightarrow -\infty} \frac{\hat{w}_{\alpha,i} \cdot p_{i,j}^\alpha}{\hat{w}_{\alpha,k} \cdot p_{k,j}^\alpha} &=&  0 \quad  \text{ if } i\notin S_j.
\end{eqnarray*}    
Therefore, there exists an $\alpha^*$ such that $\ell_i(P,\alpha^*,w_{\alpha^*}) \leq \lmks(P)+\epsilon$ for all agents $i\in [m]$. By the monotonicity property on values of $\alpha$ (first part of \Cref{thm:exponential}), for any ${\alpha'} < \alpha^*$, the value of $\ell^*(P,{\alpha'})$ is at most $\lmks(P)+\epsilon$ as required.

\end{proof}
\end{toappendix}

\section{Noise Resilience: Handling Predictions with Error}\label{sec:noise}

In this section, we show the noise resilience of our algorithms, namely that we can handle errors in the learned parameters.
First, we will show that for both objectives (\maxmin and \minmax), an  $\eta$-approximate set of learned parameters yields an online algorithm with a competitive ratio of at least/at most $\eta$.
Second, for the \minmax objective, we show that it is possible to improve the competitive ratio further in the following sense: using a set of learned parameters with a multiplicative error of $\eta$ with respect to the optimal parameters, we can obtain a $O(\log \eta)$-competitive algorithm. (This was previously shown by Lattanzi {\em et al.}~\cite{LattanziLMV20} but only for the special case of restricted assignment.) We also rule out a similar guarantee for the \maxmin objective, i.e., we show that using $\eta$-approximate learned parameters, an algorithm cannot hope to obtain a competitive ratio better than $\eta/c$ for some constant $c$. Finally, we show that noise-resilient bounds can be obtained not just for the \minmax and \maxmin objectives but also for any homogeneous monotone minimization or maximization objective function.

Formally, a weight vector $\bw$ is $\eta$-approximate with respect to a weight vector to $\bw^*$, if for any two agents $i,i'\in [m]$,
$\frac{w_{i'}}{w_i} \leq  \eta \cdot \frac{w^*_{i'}}{w^*_i}$.
First, we show a basic noise resilience property that holds for both the \minmax and \maxmin objectives: 

\begin{lemma}
\label{lem:subaddF}
Fix a weight matrix $P\in \mps$ and a transformation matrix $G \in \mps$. 
For any two parameter vectors $\bw^*, \bw\in \vps$, such that $\bw$ is $\eta$-approximate to $\bw^*$,  we have that for any agent~$k$:
\[
\frac{\ell_k(P,G, \bw^*)}{\eta}  \leq \ell_k(P,G, \bw)\leq \eta \cdot \ell_k(P,G, \bw^*).
\]    
\end{lemma}
\begin{proof}
Let $y_{i,j}=x_{i,j}(G,\bw^*)$ and $z_{i,j} = x_{i,j}(G,\bw)$ be the respective fractional allocations under proportional allocation using the transformation matrix $G$. 
For an agent $i$, let $\tau_i = w_i/w^*_i$. Then for any two agents $i,k$, we have that $1/\eta \leq \tau_k/\tau_i \leq \eta$.
\begin{toappendix}
By \Cref{obv:xval} , we have:
\end{toappendix}
\begin{noappendix}
We have,
\end{noappendix}
$\frac{y_{i,j}}{z_{i,j}} = \sum_{i'\in[m]} \frac{\tau_{i'}}{\tau_i} \cdot y_{i',j}.$
Therefore,
$$\frac{y_{i,j}}{z_{i,j}} = \sum_{i'\in[m]} \frac{\tau_{i'}}{\tau_i} \cdot y_{i',j} \geq \sum_{i'\in[m]} \frac{1}{\eta} \cdot y_{i',j} =\frac{1}{\eta} \cdot \sum_{i'\in[m]}  y_{i',j}  = \frac{1}{\eta}, \text{ and}$$
$$\frac{y_{i,j}}{z_{i,j}} = \sum_{i'\in[m]} \frac{\tau_{i'}}{\tau_i} \cdot y_{i',j} \leq \sum_{i'\in[m]} {\eta} \cdot y_{i',j} = {\eta} \cdot \sum_{i'\in[m]}  y_{i',j}  = \eta.$$
Hence, $y_{i,j}/\eta  \leq z_{i,j} \leq y_{i,j} \cdot \eta$. Finally, the lemma hold by summing over all items.
\end{proof}


The next theorem follows immediately by using a proportional allocation according to the parameter vector $\tilde{\bw}$:

\begin{theorem}
\label{thm:noise-basic}
Fix any $P,G \in \mps$. Let $\bw$ be a learned parameter vector that gives a solution of value $\gamma$ for the \maxmin (resp., \minmax) objective using proportional allocation. Let $\tilde{\bw}$ be $\eta$-approximate to $\bw$ for some $\eta > 1$. Then, there exists an online algorithm that given $\tilde{\bw}$ generates a solution with value at least $\Omega(\gamma/\eta)$ (resp., at most $O(\eta\gamma)$).
\end{theorem}

In particular, if $\bw$ is the {\em optimal} learned parameter vector in the above theorem and $\tilde{\bw}$ is an $\eta$-approximation to it, then we obtain a competitive ratio of $\Omega(1/\eta)$.

\smallskip
The rest of this section focuses on the \minmax objective for which we can obtain an improved bound. In the next lemma, we establish an upper bound on the load, using \Cref{lem:subaddF} and monotonicity. 

\begin{lemma}
\label{lem:subadd}
Fix a weight matrix $P\in \mps$ and a transformation matrix $G \in \mps$. For any two parameter vectors $\bw^*, \bw\in \vps$ such that there exists an agent $k\in [m]$ for which $w^*_k/2 \leq w_k \leq w^*_k$ and for all other agents $i\not= k$, we have $w_i\geq w^*_i/2$, then the following holds:
$\ell_k(P,G, \bw)\leq 2\cdot \ell_k(P,G, \bw^*).$
\end{lemma}

\begin{proof}
Define $\bw'$ where $w'_k = w^*_k$ (i.e., the maximum in its allowed range) and $w'_i = w^*_i/2$ for all $i\neq k$ (i.e., the minimum in their allowed ranges). Now, by monotonicity (Observation~\ref{obv:monotone}), we have
$x_{k,j}(G,\bw) \leq x_{k,j}(G,\bw')$, and therefore, $\ell_k(P,G,\bw) \leq \ell_k(P,G,\bw')$. 
Note that for $\bw'$, for any two agents $i_1,i_2$, $\frac{w_{i_1}}{w_{i_2}} \leq  2 \cdot \frac{w^*_{i_1}}{w^*_{i_2}}$. Therefore, by \Cref{lem:subaddF}, we have $\ell_k(P,G,\bw') \leq 2 \cdot \ell_k(P,G,\bw^*)$. By combining the two inequalities, we have 
$\ell_k(P,G,\bw) \leq \ell_k(P,G,\bw') \leq 2\cdot \ell_k(P,G,\bw^*)$, 
as required.
\end{proof}

\begin{algorithm}[t]
\begin{itemize}
    \item Let $\hat{\bw}$ a prediction vector and $T$ is the offline optimal objective for the \minmax problem.
    \item Initialize: $\ell_i \leftarrow 0$ and $\tilde{w}_i \leftarrow \hat{w}_i$, for all $i\in [m]$  
\end{itemize}
For each item $j$:
    \begin{itemize}
    \item Compute $x_{i,j} = \frac{f(p_{i,j}) \cdot \tilde{w}_i}{\displaystyle\sum_{i'\in [m]}f(p_{i',j})\cdot \tilde{w}_{i'}}$
    \item $\ell_i \leftarrow \ell_i + p_{i.j}\cdot x_{i,j}$, for all $i\in [m]$
    \item If exists $i\in[m]$, s.t. $\ell_i> 2\cdot T$
    \subitem Set $\ell_i \leftarrow 0$
    \subitem Update $\tilde{w}_i \leftarrow \tilde{w}_i/2$
\end{itemize}
\caption{The online algorithm with predictions.}
\label{alg:online}
\end{algorithm}

Let us denote the predicted learned parameter vector that is given offline to the \minmax algorithm by $\hat{\bw}$. We also assume that the algorithm knows the optimal objective value $T$. 
By scaling, we assume w.l.o.g that $\tilde{\bw}$ is coordinate-wise larger than the optimal learned parameter vector $\bw$. 
The algorithm uses a learned parameter vector $\hat{\bw}$ that is iteratively refined, starting with $\hat{\bw} = \tilde{\bw}$ (see \Cref{alg:online}). In each iteration, the current parameter vector $\hat{\bw}$ is used to determine the assignment using proportional allocation until an agent's load in the current phase exceeds $2T$. If this happens for any agent $i$, then the algorithm halves the value of $\hat{w}_i$, starts a new phase for agent $i$, and continues doing proportional allocation with the updated learned parameter vector $\hat{\bw}$.

\begin{theorem}
Fix any $P,G \in \mps$. Let $\bw$ be a learned parameter vector that gives a fractional solution with maximum load $T$ using proportional allocation. 
Let $\tilde{\bw}$ be an $\eta$-approximate prediction for $\bw$. Then there exists an online algorithm that given $\tilde{\bw}$ generates a fractional assignment of items to agents with maximum load at most $O(T \log \eta)$.
\end{theorem}
\begin{proof}
By the algorithm's definition, an agent's total load is at most $2T$ times the number of phases for the agent.
We show that for any agent $i$, the parameter $\tilde{w}_i$ is always at least $w_i/2$. This immediately implies that the number of phases for machine $i$ is $O(\log \eta)$, which in turn establishes the theorem.

Suppose, for contradiction, in some phase for agent $k$, we have $\tilde{w}_k < w_k/2$. Moreover, assume w.l.o.g. that agent $k$ is the first agent for which this happens.
Clearly, by the algorithm definition,
there is a preceding phase for agent $k$ when $\tilde{w}_k < w_k$.
Note that, in this entire preceding phase, we have $w_k > \tilde{w}_k \geq w_k/2$, and for all $i\neq k$, $\tilde{w}_{i} \geq w_i/2$ (by our assumption that $k$ is the first agent to have a violation). 
However, by \Cref{lem:subadd}, the load of agent $k$ in the preceding phase would be at most $2T$. This contradicts the fact that the algorithm started a new phase for agent $k$ when its load exceeded $2T$ in the preceding phase.
\end{proof}

\begin{toappendix}
We now show that the bounds obtained above for the \maxmin and \minmax objectives are asymptotically tight. 
\end{toappendix}

\begin{noappendix}
In the full version of the paper, we show that the bounds obtained above for the \maxmin and \minmax objectives are asymptotically tight. 
\end{noappendix}

\begin{toappendix}

\begin{lemma}
\label{lem:maxminlower}
There exists an instance $P$ and learned parameter vectors $\bw, \bw^*$, where $\bw$ is $\tau$-approximate with respect to $\bw^*$, such that using proportional allocation with $\bw^*$ obtains a \maxmin objective of $\Omega(\tau)$, while even when $\bw$ is known offline, any online algorithm achieves a \maxmin objective of $O(1)$.
\end{lemma}
\begin{proof}
Our construction is in the restricted assignment setting. To define $\bw$, set $w_i = 1$ for all $i\in m$. The first batch has $m$ items, where $p_{i,j} = 1$ for all $i, j\in [m]$.
Clearly, in any assignment of these items, there exists an agent $k$ such that their load at the end of the first batch is at most $1$.
The second batch consists of $(m-1)\cdot m$ items such that for $j \in \{m+1, \ldots, (m-1)\cdot m\}$, we have $p_{i,j}=1$ for $i\neq k$ and $p_{k,j}=0$. Clearly, the load of agent $k$ at the end of the second batch remains unchanged at $\le 1$, which means the \maxmin objective is also $\le 1$. (This can also be extended to randomized algorithms but choosing $k$ uniformly at random in the second batch, and using Yao's minmax principle.)

Now, define $\bw^*$ as $w^*_i = 1$ for $i\neq k$ and $w^*_k=\tau$. Then, for $m \geq \tau \geq 1$, using $\bw^*$ gives a proportional allocation with a \maxmin objective of $\Omega(\tau)$. 
\end{proof}

\begin{lemma}
\label{lem:minmaxlower}
There exists an instance $P$ and learned parameter vectors $\bw, \bw^*$, where $\bw$ is $\tau$-approximate with respect to $\bw^*$,
such that using proportional allocation with $\bw^*$ obtains a \minmax objective of $O(1)$, while even even when $\bw$ is known offline, any online algorithm achieves a \minmax objective of $\Omega(\log \tau)$.
\end{lemma}

\end{toappendix}

\begin{toappendix}

\begin{proof}
Our construction is in the restricted assignment setting and is essentially equivalent to the $\Omega(\log m)$ lower bound for the \minmax problem in the worst-case setting. To define $\bw$, set $w_i = 1$ for all $i\in m$. The example consists of $m = 2^k$ agents and $n=m-1$ items. The first batch comprises $m/2$ items, each of which has a weight of $1$ for a disjoint pair of agents, and $\infty$ for the remaining agents.
The second batch comprises $m/4$ items, each of which has a weight of $1$ for a disjoint pair of agents, and $\infty$ for the remaining agents. Crucially, for every pair of agents with finite weight for an item in the first batch, one must have load at least $\nicefrac12$ after  the first batch; this agent has finite weight for one of the items in the second batch and the other agent has infinite weights for all items in the second batch (and all batches henceforth). We continue in this manner, pruning the number of items by a factor of $2$ in every step and ensuring that the agents that have finite weight for any item in the $t$th batch must have a total load of at least $\frac{t-1}{2}$ from the previous $t-1$ batches. Clearly, the \minmax objective at the end of the algorithm is $\Omega(\log m)$. (This is true even if we allow randomized algorithms by uniformly randomizing the choice of agent to retain in any batch, and using Yao's minmax principle.)

Now, set $\tau = m$ and define $\bw^*$ as follows: $w^*_k = 2^{-a_k}$, where $a_k$ is the number of items that have a finite weight for agent $k$. A proportional allocation using these learned parameters achieves a makespan of at most~$2$.
\end{proof}
\end{toappendix}

\section{Learnability of the Parameters}\label{sec:learning}

\newcommand{\me}{\frac{m}{\epsilon}}
\newcommand{\mei}{\frac{\epsilon}{m}}
\newcommand{\met}{\frac{m^3}{\epsilon^2}}
\newcommand{\meti}{\frac{\epsilon^2}{m^3}}
\newcommand{\amin}{a_{\min}}
\newcommand{\sw}{\delta}
\newcommand{\dec}{\textbf{NET}}
\newcommand{\Ds}{\mathcal{D}}
\newcommand{\Es}{\mathbb{E}}


We consider the learning model introduced by \cite{LavastidaMRX21a}, and show that under this model, the parameter vector $\bw$ can be learned efficiently from sampled instances.
Specifically, we consider the following model: the $j$th item (i.e., the values of $\bp_j = (p_{i, j}: i\in [m])$ is independently sampled from a (discrete) distribution $\Ds_j$. In other words, the matrix $P$ of utilities is sampled from $\Ds = \times_j \Ds_j$. 

We set up the model for the \maxmin objective; the setup for the \minmax objective is very similar and is omitted for brevity. Let $T = \Es_{P\sim \Ds} [\lsan(P)]$ be the expected value of the \maxmin objective in the optimal solution for an instance $\lsan(P)$ drawn from $\Ds$. Morally, we would like to say that we can obtain a vector $\bw$ that gives a nearly optimal solution (in expectation) using proportional allocation (i.e., a \maxmin objective of $(1-\epsilon)\cdot T$ in expectation for some error parameter $\epsilon$) using a bounded (as a function of $\epsilon$) number of samples. 
Similar to \cite{LavastidaMRX21a}, we need the following assumption:

\smallskip\noindent{\bf Small Items Assumption:} Conceptually, this assumption states that each individual item has a small utility compared to the overall utility of any agent in an optimal solution. Precisely, we need $p_{i,j} \leq \frac{T}{\zeta}$ for every $i\in [m], j\in [n]$ for some value $\zeta = \Theta\left(\frac{\log m}{\epsilon^2}\right)$.

Our main theorem in this section for the \maxmin and \minmax objectives are:

\begin{theorem}
\label{thm:learning}
Fix an $\epsilon > 0$ for which the small items assumption holds. Then, there is an (learning) algorithm that samples $O(\frac{m}{\log m} \cdot \log \me)$ independent instances from $\Ds$ and outputs (with high probability) a prediction vector $\bw$ such that using $\bw$ in the proportional allocation scheme gives a \maxmin objective of at least $(1-\Omega(\epsilon))\cdot T$ in expectation over instances $P\sim \Ds$.
\end{theorem}


\begin{theorem}
\label{thm:learning-minmax}
Fix an $\epsilon > 0$ for which the small items assumption holds. Then, there is an (learning) algorithm that samples $O(\frac{m}{\log m} \cdot \log \me)$ independent instances from $\Ds$ and outputs (with high probability) a prediction vector $\bw$ such that using $\bw$ in the proportional allocation scheme gives a \minmax objective of at most $(1+O(\epsilon)) T$ in expectation over instances~$P\sim \Ds$.
\end{theorem}

Importantly, the description of the entries of $\bw$ in \Cref{thm:learning} and \Cref{thm:learning-minmax} are bounded. Specifically, let us define $\dec(m,\epsilon) \subseteq \vps$ as follows: (a) for the \maxmin objective, $\bw\in \dec(m,\epsilon)$ if there exist vectors $\mathbf{u},\mathbf{\sw} \in \vps$ such that $w_{i} = \frac{\sw_i}{u_i^\alpha}$ and $u_i,\sw_i \in \left\{\left(\frac{1}{1-\epsilon}\right)^r:  r \in [K]\right\}$ for some $K = O(\me \log \me)$, and (b) for the \minmax objective, $\bw\in \dec'(m,\epsilon)$ if there exist vectors $\mathbf{u},\mathbf{\sw} \in \vps$ such that $w_{i} = \frac{\sw_i}{u_i^\alpha}$ and $u_i,\sw_i \in \left\{(1+\epsilon)^r:  r \in [K]\right\}$ for some $K = O(\me \log \me)$. The vectors $\bw$ produced by the learning algorithm in \Cref{thm:learning} and \Cref{thm:learning-minmax} will satisfy $\bw\in \dec(m, \epsilon)$ and $\bw\in \dec'(m, \epsilon)$ in the respective cases.



\smallskip\noindent{\bf Proof Idea for \Cref{thm:learning} and \Cref{thm:learning-minmax}.}
Recall that in PAC theory, the number of samples needed to learn a function from a family of $N$ functions is about $O(\log N)$. Indeed, restricting $\bw$ to be in the class $\dec(m, \epsilon)$ or $\dec'(m, \epsilon)$ serves this role of limiting the hypothesis class to a finite, bounded set since $|\dec(m, \epsilon)| = |\dec'(m, \epsilon)| = K^{2m}$ where $K = O(\me \log \me)$. Using standard PAC theory, this implies that using about $O(m \log K) = O(m \cdot \log \frac{m}{\epsilon})$ samples, we can learn the ``best'' vector in $\dec(m, \epsilon)$ or $\dec'(m, \epsilon)$ depending on whether we have the \maxmin or \minmax objective. Our main technical work is to show that this ``best'' vector produces an approximately optimal solution when used in proportional allocation. We state this lemma next:

\begin{lemma}
\label{lem:limited}
Fix any $P$. For the \maxmin objective, there exists a learned parameter vector $\bw\in \dec(m,\epsilon)$ which when used in \ep-allocation gives a $1-\Omega(\epsilon)$ approximation. For the \maxmin objective, there exists a learned parameter vector $\bw'\in \dec'(m,\epsilon)$ which when used in \ep-allocation gives a $1+O(\epsilon)$ approximation.  
\end{lemma}

\begin{toappendix}
\subsection{Proof of \Cref{lem:limited}}\label{sec:proof-limited}

\subsubsection{Preprocessing: Modification of $P$}
We will not show \Cref{lem:limited} directly on an arbitrary matrix $P$. Instead, we will first ``preprocess'' $P$ to establish some properties that will help us show \Cref{lem:limited}. 

The first step performs discretization. For the \maxmin objective, we round down each $p_{i, j}$ value to an integer power of $\frac{1}{1-\epsilon}$. This changes the optimal \maxmin objective by at most $1-\epsilon$. Similarly, for the \minmax objective, we round up each $p_{i, j}$ value to an integer power $1+\epsilon$. This changes the optimal \minmax objective by at most a factor of $1+\epsilon$.

In the second step, the goal is to ensure that the ratio between any two entries $p_{i, j}$ and $p_{k, j}$ is bounded. For the \minmax objective, this is simple: if $\frac{p_{i,j}}{p_{k,j}} > \me$ for some $k\in [m]$, then we set $p_{i,j} = \infty$, i.e., $x_{i, j} = 0$. This transformation increases the \minmax objective by at most a factor of $1+\e$. 

For the \maxmin objective, the second step is more complicated. We modify the online allocation algorithm to assign an $\frac{\epsilon}{m}$ fraction of each item to every agent. Since we still have a $1-\epsilon$ fraction of every item left, this step changes the optimal \maxmin objective by at most a factor of $1-\epsilon$. But, what does this allocation of $\e$-fraction of each item achieve? Let $a_i = \sum_{j\in [n]} p_{i, j}$ denote the {\em monopolist value} of agent $i\in [m]$, i.e., the total utility if all items were assigned to agent $i$. We assume that the values of $a_i$ for all $i\in [m]$ are known to the algorithm -- in fact, these values can also be learned to sufficient accuracy but we ignore this additional learning step for simplicity and assume these values are known. Now, note that $\lsan \le \amin$ where $\amin$ is defined as $\min_{i\in[m]} a_i$.  The allocation of $\e$ fraction of every item ensures that every agent $i\in [m]$ with a {\em large} monopolist value satisfying $a_i \ge \frac{m}{\epsilon}\cdot \amin$ gets a load of at least $\frac{\epsilon}{m}\cdot a_i \ge \amin \ge \lsan$ just from this $\e$-allocation. Therefore, we ignore these agents in the rest of the analysis and assume $a_i < \frac{m}{\epsilon} \cdot \amin$ for every agent $i\in [m]$.

Now, fix any agent $i\in [m]$ and define $J_i$ to be the set of items $j\in [m]$ for each of which there exists another agent $k(j)$ such that $\frac{p_{k,j}}{p_{i,j}} > \met$. Then, we set $p_{i,j} = 0$.
This modification decreases the optimal \maxmin objective by a factor of at most $1-\epsilon$ because:
\[
    \sum_{j\in J_i} p_{i, j} 
    \le \frac{\sum_{j\in J_i} p_{k(j), j}}{\met}
    \le \frac{\sum_k a_k}{\met} 
    \le^{\rm (by~first~preprocessing~step)} \frac{m\cdot \frac{m}{\epsilon} \cdot \amin}{\met}
    = \epsilon\cdot \frac{\amin}{m}
    \le \epsilon\cdot \lsan,
\]    
where the last step follows from $\lsan \ge \frac{\amin}{m}$ by a uniform assignment.

So, in essence, we can assume for both the \minmax and \maxmin objectives, the following holds for any item $j\in [n]$: if $x^*_{i, j}, x^*_{k, j} \not= 0$ in an optimal solution $x^*$, then we can assume that $\frac{p_{i, j}}{p_{k, j}} \le \poly(m/\e)$.

\subsubsection{Proof of \Cref{lem:limited} for the \maxmin objective}
We now prove \Cref{lem:limited} for the \maxmin objective. The proof for the \minmax objective is similar, and we omit it for brevity.

Recall that to define any vector $\bw\in \dec(m, \e)$, we need to define two vectors $\bu$ and $\delta$. We define these vectors for the vector $\bw$ in \Cref{lem:limited} separately in the next two subsections. Note that \Cref{lem:limited} is existential; hence, we can use the optimal solution, for instance, in the proof. 

\smallskip\noindent{\bf The vector $\bu$.}
Given a preprocessed matrix $P\in \vps$ and an optimal solution $x^*$ for the \maxmin objective, we define an auxiliary directed graph $G_{x^*}(V,E)$ as follows:
\begin{itemize}
    \item The set of vertices $V = [m]\cup \{0\}$, i.e., the agents and a special vertex labeled $0$.
    \item The set of edges $E = ([m]\times [m]) \cup (\{0\} \times [m])$, i.e., all edges between (ordered) pairs of vertices representing the agents (including self loops) and edges from the special vertex to all the vertices representing the agents. Note that the set of vertices and edges does not depend on $x^*$.
    \item We now define a cost function on the edges that does depend on $x^*$. Edges in $[m]\times [m]$ have the following costs:
    \[ 
        c_{i,k} = \ln \left(\max_{j\in [n]}\left\{\frac{p_{k,j}}{p_{i,j}} \;\bigg\vert\; x^*_{i,j}>0\right\}\right).
    \]
    In other words, the cost of an edge $(i, k)$ is the logarithm of the maximum ratio of the weight of an item for $k$ to that for $i$ among those items that have a non-zero allocation to agent $i$ in $x^*$.
    In addition, all edges incident on the special vertex have cost $0$, i.e., $c_{0,k} = 0$ for all $k\in[m]$.
\end{itemize}

Similar to \Cref{lem:negweight}, one can verify that $G_{x^*}$ does not contain a negative cycle; if not, one can compute a different assignment in which the load of some agent increases while keeping all other agents at the same load.

\begin{lemma}
\label{lem:negweightSan}
    Given a processing matrix $P$ and an optimal solution $x^*$ resulting in an objective value of $\lsan(P)$ for the \maxmin problem, the auxiliary graph $G_{x^*}$ does not contain a negative cycle.
\end{lemma}

For any agent $i\in [m]$, \Cref{lem:negweightSan} allows us to define $c^*_i$ as the minimum cost of a path from vertex $0$ to vertex $i$ in the auxiliary graph $G_{x^*}$.
We now define a {\em ratio vector} $\bu \in \vps$, where $u_{i} = e^{c^*_{i}}$.
As in \Cref{lem:restrcitedtransformation}, one can show that:

\begin{lemma}
\label{lem:restrcitedtransformationSan}
Given a matrix $P$ and an optimal solution $x^*$ resulting in a \maxmin objective of $\lsan(P)$, we have that any $i\in [m]$, $j\in [n]$, if there exists some agent $k\in [m]$ such that $\frac{p_{k,j}}{u_k} > \frac{p_{i,j}}{u_i}$, then $x^*_{i,j} = 0$.
\end{lemma}
We also note the following property of $\bu$ that follows immediately from the third preprocessing step:
\begin{lemma}
\label{lem:power}
    Each coordinate of vector $\bu$ is an integer power of $\frac{1}{1-\e}$.
\end{lemma}

\smallskip\noindent{\bf Bounding the aspect ratio of the ratio vector.}
We show the following:
\begin{lemma}\label{lem:aspect}
For any $i, k\in [m]$, the aspect ratio of the ratio vector $\bu$ is bounded as follows:
\[
    \frac{u_k}{u_i} \le \left(\met\right)^m.
\]
\end{lemma}
\begin{proof}
First, note that since there is a directed edge of zero cost from vertex $0$ to every the vertex for every agent $i\in [m]$, we have
\[
u_i\le 1 \text{ for all } i\in [m].
\]
Next, we bound the minimum value of $u_i$ for any agent $i$. Recall that by preprocessing, we have the following: 
if $p_{i,j},p_{k,j}>0$, then $\frac{p_{i,j}}{p_{k,j}} \le \frac{\epsilon^2}{m^3}$ if $x^*_{i, j} > 0$. Therefore, $c_{i, k}\ge \ln \frac{\epsilon^2}{m^3}$ for all $i, k\in [m]$. Since the shortest path contains at most $m$ edges, therefore $c^*_i \geq m\cdot \ln \frac{\epsilon^2}{m^3}$, i.e., 
\[
    u_i \ge \left(\frac{\epsilon^2}{m^3}\right)^m \text{ for all } i\in [m].
\]    
We can now conclude the lemma from the upper and lower bounds on $u_i$ for all $i\in [m]$.
\end{proof}

\smallskip\noindent{\bf The vector $\sw$.}
We first define a restricted related instance of the problems based on the value of $\bu$.
For any item $j \in [n]$, let $\gamma_j = \displaystyle\max_i \frac{p_{i,j}}{u_i} $
and $S_j =  \arg\max_i \frac{p_{i,j}}{u_i}
= \left\{i \in [m]: \frac{p_{i,j}}{u_i} = {\gamma_j} \right\}$.
By \Cref{lem:restrcitedtransformationSan}, there exists an optimal solution $x^*$ such that $x_{i,j} = 0$ if $i \notin S_j$.


\Cref{lem:restrcitedtransformationSan} allows us to convert any general matrix $P$ to a restricted related instance while preserving the value of $\lsan(P)$. 

We define the following:
\begin{itemize}
    \item an utility vector $\hat{\bp}  \in \vpsn$ where $\hat{p}_j = \gamma_j$.
    \item a speed vector $\hat{\mathbf{v}} \in \vps$ where $\hat{v}_i = 1/{u_i}$. 
    \item an admissibility matrix $\hat{E} \in \{0,1\}^{m \times n}$ where $\hat{E}_{i, j} = 1$ if and only if $i\in S_j$.
\end{itemize}
Note that by \Cref{lem:restrcitedtransformationSan}, $x^*$ is a feasible solution to this restricted related instance, and produces the same load for every agent:
\[
\ell_i(x^*, \hat{\bp},\hat{v},\hat{E}) = \sum_j x^*_{i,j} \cdot \frac{\hat{p}_j}{\hat{v}_i}= \sum_{j:x^*_{i,j}>0}  x^*_{i,j}  \gamma_j u_i = \sum_j x^*_{i,j} p_{i,j} \geq \lsan(P).
\]
Conversely, let $\hat{x}$ be a solution to the restricted related instance. We have:
\[
\ell_i(\hat{x}, \hat{\bp},\hat{v},\hat{E}) 
= \sum_{j:S_j\ni i} \hat{x}_{i,j}\cdot \frac{\hat{p}_j}{\hat{v}_i}
= \sum_{j:S_j\ni i} \hat{x}_{i,j} \gamma_j {u}_i
= \sum_{j:S_j\ni i} \hat{x}_{i,j} p_{i,j} 
= \ell_i(\hat{x}, P).
\]

We now invoke \Cref{thm:sinkhorn}, which yields:
\begin{corollary}
\label{cor:sinkhorn2}
There exists a vector of parameters $\sw \in \vps$ such that the following allocation
\[
    \hat{x}_{i,j}(\sw) = \begin{cases}
                    \frac{\sw_i}{\sum_{i'\in S_j} \sw_{i'}} & \text{ if } i\in S_j \\
                    0 & \text{ otherwise }
                    \end{cases}
\]                    
for the restricted related instance achieves the optimal \maxmin objective (denoted $\lsan(\hat{\bp},\hat{v},\hat{E})$).
\end{corollary}

Next, we approximate the vector $\delta$ in \Cref{cor:sinkhorn2} with a vector  $\sw'$ with a bounded aspect ratio, and show that this approximation only loses a factor of $1-\epsilon$. In fact, we will show that
$\hat{x}_{i,j}(\sw')  \geq (1-\epsilon) \cdot \hat{x}_{i,j}(\sw)$. 

We give an algorithm for computing $\sw'$. Let $i_1,\dots i_m$ be the ordered indices in increasing order of values of the coordinates of $\sw$, i.e., $\sw_{i_1} \le \sw_{i_2} \le \ldots \le \sw_{i_m}$. Initialize $\sw'_i = \sw_i$ for all $i\in [m]$. Next, we update the values of $\sw'$ iteratively using the following rule in each iteration: for each  $k\in[m-1]$ satisfying the condition $\frac{\sw'_{i_{k+1}}}{\sw'_{i_k}} > \me$, we multiply $\sw'_{i_r}$ by $\frac{\sw'_{i_{k}}}{\sw'_{i_{k+1}}} \cdot \me$ for every $r =\{k+1,\dots, m\}$. In effect, the ratio $\frac{\sw'_{i_{k+1}}}{\sw'_{i_k}}$ becomes $\me$ and the ratios between all other pairs $\frac{\sw'_{i_{k'+1}}}{\sw'_{i_{k'}}}$ for $k' \not= k$ remains unchanged. (A similar trick also appears in \cite{LiX21}.)

The following inequality holds for every item $j\in[m]$ and any agent $i\in S_j$:
\begin{align*}
\frac{\sw'_{i}}{\sum_{k\in S_j}\sw'_{k}} 
\ge (1-\epsilon)\cdot \frac{\sw_{i}}{\sum_{k\in S_j}  \sw_{k}}
\end{align*}
\eat{

Fix a job $j$, and let $i_r$ let $S_1 = S_j \cap 
\{r,\dots,m\}$,
$S_2 = S_j \cap \{i_a: a<r, w_{i_a}/w_{i_r} \geq \mei\} $,
$S_3 = S_j \cap \{i_a: a<r, w_{i_a}/w_{i_r} < \mei\} $,
Assume by scaling that $w_{i_r} = w'_{i_r} = 1$

\begin{align*}
\frac{\dsum_{i_k\in S_j} w_{i_k}}
{\dsum_{i_k\in S_j} w'_{i_k}}
& = \frac
{1+\dsum_{i_k\in S_1} w_{i_k} + \dsum_{i_k\in S_2} w_{i_k} + \dsum_{i_k\in S_3} w_{i_k}}
{1+\dsum_{i_k\in S_1} w'_{i_k} + \dsum_{i_k\in S_2} w'_{i_k} + \dsum_{i_k\in S_3} w'_{i_k}}
\\ & \geq \frac
{1+\dsum_{i_k\in S_1} w_{i_k} + \dsum_{i_k\in S_2} w_{i_k} + \dsum_{i_k\in S_3} w_{i_k}}
{1+\dsum_{i_k\in S_1} w_{i_k} + \dsum_{i_k\in S_2} w'_{i_k} + \dsum_{i_k\in S_3} w'_{i_k}}
\\ & \geq \frac
{1+\dsum_{i_k\in S_1} w_{i_k} + \dsum_{i_k\in S_2} w_{i_k} + \dsum_{i_k\in S_3} w_{i_k}}
{1+\dsum_{i_k\in S_1} w_{i_k} + \dsum_{i_k\in S_2} w_{i_k} + \dsum_{i_k\in S_3} w'_{i_k}}
\\ & \geq \frac
{1+ \dsum_{i_k\in S_3} w_{i_k}}
{1+ \dsum_{i_k\in S_3} w'_{i_k}}
 \geq \frac{1}{1+ \dsum_{i_k\in S_3} \mei}\geq
\frac {1}{1+ \epsilon}
\end{align*}
}

By scaling, we assume that $\min_i \sw'_i = 1$. By our construction we have:
\[
    \delta'_i \le \left(\frac{m}{\epsilon}\right)^m \text{ for all } i\in [m].
\]   

Let $\tilde{\sw}$ be derived from $\sw'_i$ by rounding up to the nearest integer power of $\frac{1}{1-\epsilon}$. Then for $i\in S_j$, we have
$\tilde{x}_{i,j} \geq (1-\epsilon) x'_{i,j}$. This completes the definition of $\sw$.

\newcommand{\tsw}{\tilde{\sw}}

\smallskip\noindent{\bf The vector $\bw$ in \Cref{lem:limited}.}
Now, given such $\bu,\tilde{\sw}$, for a fixed $\alpha$, we define the vector $\bw\in \dec(m,\epsilon)$ in \Cref{lem:limited}.
Set $w_{i} = \frac{\tilde{\sw}}{u_i^\alpha}$. Then, it is sufficient to show that, for $i\in S_j$, 
$$ x_{i,j}  = \frac{w_{i} \cdot p_{i,j}^\alpha}{\sum_{i'} w_{i'} \cdot p_{i',j}^\alpha} \geq (1-\epsilon) \cdot \tilde{x}_{i,j}$$
We have
$$ \frac{w_{i} \cdot p_{i,j}^\alpha}{w_{k} \cdot p_{k,j}^\alpha} 
=
\frac{\tsw_i \cdot  p_{i,j}^\alpha / {u_i}^\alpha}{ \tsw_k \cdot p_{k,j}^\alpha/{u_k}^\alpha } 
= \frac{\tsw_i}{\tsw_k} \cdot \left( \frac{{u_k}\cdot p_{i,j}}{ {u_i} \cdot p_{k,j}} \right)^\alpha.$$
Now, fix an item $j$ and an agent $i \in S_j$. We have the following two cases for any agent $k\in [m]$ (we use $\alpha = \frac{2m}{\e} \cdot \log_{1-\epsilon} \me$):
\begin{eqnarray*}
    \frac{w_{i} \cdot p_{i,j}^\alpha}{w_{k}\cdot p_{k,j}^\alpha} &=& \frac{\tsw_i}{\tsw_k} \quad  \text{ if } k\in S_j\\ 
    \frac{w_{i} \cdot p_{i,j}^\alpha}{w_k \cdot p_{k,j}^\alpha} &\leq&  \frac{\tsw_i}{\tsw_k} \cdot (1-\epsilon)^\alpha \leq \left(\frac{m}{\epsilon}\right)^m (1-\epsilon)^\alpha \leq \frac{\epsilon}{m} \quad  \text{ if } k \notin S_j,
\end{eqnarray*} 
where the first inequality is by our construction if $i\in S_j$ and $k\notin S_j$ then $\left( \frac{{u_k}\cdot p_{i,j}}{ {u_i} \cdot p_{k,j}} \right) \leq 1-\epsilon$, the second inequality is since $1 \leq \tilde{\sw}_i \leq (\frac{m}{\epsilon})^m$, the third inequality is by $\alpha$'s definition.
Therefore, 
$$ x_{i,j}  = \frac{1}{\sum_{i'} \frac{w_{i'} \cdot p_{i',j}^\alpha}{w_{i} \cdot p_{i,j}^\alpha}} \geq  
\frac{1}{1+\sum_{i'\neq i, i\in S_j} \frac{\tsw_i}{\tsw_k}+\sum_{i\notin S_j} \frac{\epsilon }{m}}
\geq  
\frac{1}{1+\sum_{i'\neq i, i\in S_j} \frac{\tsw_i}{\tsw_k}+\epsilon }
\geq (1-\epsilon) \cdot \tilde{x}_{i,j}$$

This completes the proof of \Cref{lem:limited}.

The rest of the proof, i.e. from \Cref{lem:limited} to \Cref{thm:learning}, uses standard PAC theory and closely follows Li and Xian~\cite{LiX21}. We include it for completeness in \Cref{sec:pac}.

\subsection{PAC Learning: From \Cref{lem:limited} to \Cref{thm:learning} for the \maxmin objective}
\label{sec:pac}

Let us consider a combination of all instances in the support of the distribution $\Ds$. For $L$ processing matrices $P^{(1)},P^{(2)},\dots,P^{(L)}$.
we define $P^{\text{all}} = \bigoplus_{r=1}^L P^{(r)}$ to be the instance defined by the $n\cdot L$ items.
For every $\ell \in [L]$ and $j\in [n]$, we have an item $j^{(\ell)}$ with utility vector $\bp_j^{(\ell)}$.

The following observation is immediate (superadditivity):

\begin{observation}
\label{obv:santasub}
$\lsan(P^{\text{all}}) \geq \sum_{r=1}^{L} \lsan(P^{(r)})$.
\end{observation}

Using this observation, we can prove the following lemma, by considering the combination of all instances in $\Ds$, scaled by their respective probabilities.

\begin{lemma}
\label{lem:d4}
There exists $\bw \in \dec(m,\epsilon)$, such that for every $i\in[m]$, we have
$$\displaystyle\Es_{P\sim \Ds} [\ell_i(P,\bw)] \geq (1-\epsilon)\cdot T.$$
\end{lemma}
\begin{proof}
Consider the instance $\mathbb{P} = \bigoplus \text{Pr}_D[P]\cdot P$
where $\text{Pr}_D[P]$ is the probability mass of $P$ in $\Ds$, and 
$\text{Pr}_D[P]\cdot P$ is the matrix $P$ multiplied by $\text{Pr}_D[P]$.
By \Cref{obv:santasub}, we have
\[
\lsan(\mathbb{P}) \geq \sum_{P} \text{Pr}_{\Ds}[P]\lsan(P) = \Es_{P\sim D}[\lsan(P)]= T.
\]

We can apply \Cref{lem:limited} to the combined instance to show there exists $\bw^* \in \dec(m,\epsilon)$ such that for every $i\in [m]$, we have,
\[
\sum_j x_{i,j}(\mathbb{P},\bw^*) \cdot p_{i,j} \geq (1-\epsilon) \lsan(\mathbb{P})\geq (1-\epsilon) \cdot T
\]
where $j$ indexes over all items in $\mathbb{P}$. Notice that $x_{i,j}(\mathbb{P},\bw^*)$ depends on the utility vector for item $j$, which is part of the instance $P\in \Ds$ that $j$ belongs to.
Therefore, the left side of the above inequality is exactly

\[
\sum_P \sum_{j\in[n]} x_{i,j} (P,\bw^*) \cdot \text{Pr}_{\Ds}[P] \cdot  p_{i,j}
= \Es_{P \sim \Ds}\sum_{j\in [n]}  x_{i,j}(P,\bw^*) p_{i,j} = 
\Es_{P \sim \Ds}\sum_{j\in [n]}  \ell_i(P,\bw^*),
\]
as required.
\end{proof}

For any real numbers $A,B,\epsilon,C$, we use $A\approx_{\epsilon,C} B$ to denote
$|A-B| \leq \epsilon\cdot \max (B,C)$.
The next lemma appears in \cite{LiX21}:
\begin{lemma}[Lemma D.6 in in \cite{LiX21}]
\label{lem:d6}
For any $\bw \in \dec(m,\epsilon)$, with high probability over $P\sim \Ds$, we have
\[
\forall i\in [m] : \ell_i(P,\bw) \approx_{\epsilon,T} \Es_{P'\sim D}\ell_i(P',\bw).
\]
\end{lemma}

\noindent\textbf{The learning algorithm.}
We sample $H = O \left( \frac{m}{\log m} \log \me \right)$ instances $P_1,P_2,\dots, P_H$ independently and randomly form $\Ds$. We output $\tilde{\bw} \in \dec(m,\epsilon)$ that maximizes $\min_{i\in[m]} \frac{1}{H} \sum_{h=1}^H \ell_i(P_h, \tilde{\bw})$.

The next lemma also appears in \cite{LiX21}:
\begin{lemma}[Lemma D.7 in \cite{LiX21}]
\label{lem:d7}
With probability at least $1-\frac{1}{K^m}$, for every $\bw \in \dec(m,\epsilon)$ and for every $i\in [m]$, we have
\[
    \frac{1}{H} \sum_{h=1}^H \ell_i(P_h,\bw) \approx_{\epsilon,T} \Es_{P\sim D}\ell_i(P,\bw).
\]
\end{lemma}
Now assume the event in \Cref{lem:d7} happens. Then by \Cref{lem:d4}, there exists some $\bw\in \dec(m, \e)$ such that
\[
    \min_{i\in[m]} \frac{1}{H} \sum_{h=1}^H \ell_i(P_h,\bw) \geq (1-\epsilon)^2 \cdot T.
\]    
In particular, since $\tilde{\bw}$ maximizes $\min_{i\in[m]} \frac{1}{H} \sum_{h=1}^H \ell_i(P_h,\tilde{\bw})$
for $\tilde{\bw}\in \dec(m, \e)$, we can conclude that
\[
    \min_{i\in[m]} \frac{1}{H} \sum_{h=1}^H \ell_i(P_h,\tilde{\bw}) \geq (1-\epsilon)^2 \cdot T.
\]
Applying \Cref{lem:d7} again, we get
\[
    \min_{i\in [m]} \Es_{P\sim \Ds} \ell_i(P,\tilde{\bw}) \geq (1-\epsilon)^3 \cdot T.
\]    
We now apply \Cref{lem:d6} to $\tilde{\bw}$. We have that with high probability over $P\sim \Ds$, for every $i\in [m]$ the following holds:
\[ 
\ell_i(P,\tilde{\bw}) \geq \Es_{P'\sim \Ds} \ell_i(P',\tilde{\bw})  -\epsilon \cdot \max\{T, \Es_{P'\sim \Ds} \ell_i(P',\tilde{\bw})\} \geq (1-\epsilon)^4\cdot T.
\]
Therefore, $\lsan(P,\tilde{\bw}) \geq (1-\Omega(\epsilon))\cdot T$.
This completes the proof of \Cref{thm:learning}.

\end{toappendix}

\section{Generalization to Well-Behaved Objectives}\label{sec:general}
We first generalize \Cref{thm:epa} to all well-behaved functions.
\begin{theorem}
\label{thm:general}
    Fix any instance of an online allocation problem with divisible items where the goal is to maximize or minimize a monotone homogeneous objective function. Then, there exists an online algorithm and a learned parameter vector in $\vps$ that achieves a competitive ratio of~$1-\epsilon$ (for maximization) or $1+\epsilon$ (for minimization).
\end{theorem}

\begin{proof}
Fix an objection function $f$ and a matrix $P\in \mps$. Let $\ell^f_i$ denote the load of agent $i$ in an optimal solution for objective function $f$. Also, let $x_{i, j}$ denote the fraction of item $j$ assigned to agent $i$ in this optimal solution.
Now, consider the matrix $\tilde{P}$, where $\tilde{p}_{i,j} = \frac{p_{i,j}}{\ell^f_i}$.
By the monotonicity property of $f$, the optimal objective value for $\tilde{P}$ is $1$. Therefore,
by \Cref{thm:exponential}, there exist $\alpha$ and $\tilde{\bw}$,
such that using an \ep-allocation, we get 
$\ell^*(\tilde{P}, \alpha, \tilde{\bw}) \ge 1-\epsilon$ for maximization and
$\ell^*(\tilde{P}, \alpha, \tilde{\bw}) \le 1+\epsilon$ for minimization.
Let $x^*_{i, j}$ be the fraction of item $j$ assigned to agent $i$ in this approximate solution.
By the definition of \ep-allocation, $x^*_{i,j}$ is proportional to $\tilde{p}_{i,j}^{\alpha} \cdot \tilde{w}_i 
= \left(\frac{p_{i,j}}{\ell^f_i}\right)^{\alpha} \cdot \tilde{w}_i
= p_{i,j}^\alpha\cdot \frac{\tilde{w}_i}{{(\ell^f_i)}^{\alpha}}$. Thus, if we define $\bw$ such that $w_i  = \frac{\tilde{w}_i}{{(\ell^f_i)}^{\alpha}}$, then the corresponding \ep-allocation gives a $(1-\e)$-approximate solution for maximization and $(1+\e)$-approximate solution for minimization.
\end{proof}

\subsection{Noise Resilience}
Next, we consider noise resilience for well-behaved functions, i.e., we generalize \Cref{thm:noise-basic} to all well-behaved objective functions. 
This follows immediately from \Cref{lem:subaddF} and the observation that if all loads are scaled by $\eta$, then the objective value for a well-behaved objective is also scaled by $\eta$. We state this generalized theorem below:

\begin{theorem}
Fix any $P,G \in \mps$ and any monotone, homogeneous function $f$. Let $\bw$ be a learned parameter vector that gives a solution of objective value $\gamma$ using \ep-allocation. Let $\tilde{\bw}$ be $\eta$-approximate to $\bw$ for some $\eta > 1$. Then, the \ep-allocation for  $\tilde{\bw}$ gives a solution with value at least $\gamma/\eta$ for maximization and at most $\eta\gamma$ for minimization. 
\end{theorem}
\subsection{Learnability}

Finally, we consider learnability of parameters for well-behaved functions, i.e., we generalize \Cref{thm:learning} and by assuming additional property of the objective function:
\begin{itemize}
    \item For a maximization objective $f$, we need {\em superadditivity}: $f(\sum_r \ell_r) \ge \sum_r f(\ell_r)$.
    \item For a minimization objective $f$, we need {\em subadditivity}: $f(\sum_r \ell_r) \le \sum_r f(\ell_r)$.
\end{itemize}
\begin{noappendix}

\end{noappendix}


\begin{theorem}
\label{thm:learningGen}
Let $f$ be a well-behaved function. If $f$ is superadditive, the following theorem holds for maximization of $f$, while if $f$ is subadditive, the following theorem holds for minimization of $f$. Let $T$ be the expectation of the maximum value of $f$ over instances sampled from $\Ds$. Fix an $\epsilon > 0$ for which the small items assumption holds. Then, there is an (learning) algorithm that samples $O(\frac{m}{\log m} \cdot \log \me)$ independent instances from $\Ds$ and outputs (with high probability) a prediction vector $\bw$ such that using $\bw$ in the \ep-allocation gives a value of $f$ that is at least $(1-\Omega(\epsilon))\cdot T$ for maximization and at most $(1+O(\epsilon))\cdot T$ for minimization, in expectation over instances $P\sim \Ds$.
\end{theorem}
\begin{toappendix}
\begin{proof}
Fix a maximization objective function $f$ and distribution $\Ds$ (the proof for a minimization objective is similar and omitted for brevity). 
Consider the instance $\mathbb{P} = \bigoplus \text{Pr}_D[P]\cdot P$
where $\text{Pr}_D[P]$ is the probability mass of $P$ in $\Ds$, and 
$\text{Pr}_D[P]\cdot P$ is the matrix $P$ multiplied by $\text{Pr}_D[P]$.
By our superadditivity assumption, we have \[
f(\ell^f(\mathbb{P})) \geq \sum_{P} \text{Pr}_{\Ds}[P]f(\ell^f(P)) = \Es_{P\sim D}[f(\ell^f(P))]= T.
\]

Suppose we sample $H = O \left( \frac{m}{\log m} \log \me \right)$ instances $P^{(1)},P^{(2)},\dots, P^{(H)}$ independently and randomly from $\Ds$.
Now, using the small items assumption, it is possible to compute $\hat{\ell}^f_i$ which is a $(1+\epsilon)$ approximation to $\ell^f_i(\mathbb{P})$ for all $i\in [m]$.
Similar to the previous construction, given a matrix $P$, we define $\tilde{P}$ as $\tilde{p}_{i,j} = \frac{{P}_{i,j}}{\hat{\ell}^f_i}$.
By the monotonicity property of $f$, we have: 
$\Es_{P\sim D}[\lsan(\tilde{P})]\geq 1-\epsilon$.


We output $\bw^* \in \dec(m,\epsilon)$ that maximizes $\min_{i\in[m]} \frac{1}{H} \sum_{h=1}^H \ell_i(\tilde{P}_h, \tilde{\bw})$. 
Then according to the proof of \Cref{thm:learning}, for $P\sim D$, we have with high probability for every $i\in [m]$:
\[
\ell_i(\tilde{P},\bw^* ) \geq 1-\Omega(\epsilon).
\]

Let us now define $\bw$ such that $w_i  = \frac{w^*_i}{({\hat{\ell}^f_i)}^{\alpha}}$. Then, by the homogeneity property, for a random $P\sim \Ds$, the objective function $f$ corresponding to the assignment
 $x_{i,j}(P,\bw)$ is at least $(1-\Omega(\epsilon)) \cdot T$ with high probability.
\end{proof}
\end{toappendix}

\section{Conclusion and Future Directions}\label{sec:conclusion}
In this paper, we gave a unifying framework for designing near-optimal algorithm for fractional allocation problems for essentially all well-studied minimization and maximization objectives in the literature. The existence of this overarching framework is rather surprising because the corresponding worst-case problems exhibit a wide range of behavior in terms of the best competitive ratio achievable, as well as the techniques required to achieve those bounds. It would be interesting to gain further understanding of the optimal learned parameters introduced in this paper. One natural conjecture is that these are optimal dual variables for a suitably defined convex program (for instance, such convex programs are known for restricted assignment and $b$-matching~\cite{agrawal2018proportional}). Another interesting direction of future work would be to explore other polytopes beyond the simple assignment polytope considered in this paper, such as that corresponding to congestion minimization problems.

\bibliographystyle{alpha} 
\bibliography{refs}

\newcommand{\etalchar}[1]{$^{#1}$}
\begin{thebibliography}{LMRX21b}

\bibitem[AAF{\etalchar{+}}97]{AspnesAFPW97}
James Aspnes, Yossi Azar, Amos Fiat, Serge~A. Plotkin, and Orli Waarts.
\newblock On-line routing of virtual circuits with applications to load
  balancing and machine scheduling.
\newblock {\em J. {ACM}}, 44(3):486--504, 1997.

\bibitem[AAG{\etalchar{+}}95]{AwerbuchAGKKV95}
Baruch Awerbuch, Yossi Azar, Edward~F. Grove, Ming{-}Yang Kao, P.~Krishnan, and
  Jeffrey~Scott Vitter.
\newblock Load balancing in the l\({}_{\mbox{p}}\) norm.
\newblock In {\em 36th Annual Symposium on Foundations of Computer Science},
  pages 383--391. {IEEE} Computer Society, 1995.

\bibitem[AGKK20]{AntoniadisGKK20}
Antonios Antoniadis, Themis Gouleakis, Pieter Kleer, and Pavel Kolev.
\newblock Secretary and online matching problems with machine learned advice.
\newblock In {\em Advances in Neural Information Processing Systems 33: Annual
  Conference on Neural Information Processing Systems 2020, NeurIPS 2020},
  2020.

\bibitem[ALT21]{AzarLT21}
Yossi Azar, Stefano Leonardi, and Noam Touitou.
\newblock Flow time scheduling with uncertain processing time.
\newblock In {\em {STOC} '21: 53rd Annual {ACM} {SIGACT} Symposium on Theory of
  Computing}, pages 1070--1080. {ACM}, 2021.

\bibitem[ALT22]{AzarLT22}
Yossi Azar, Stefano Leonardi, and Noam Touitou.
\newblock Distortion-oblivious algorithms for minimizing flow time.
\newblock In {\em Proceedings of the 2022 {ACM-SIAM} Symposium on Discrete
  Algorithms, {SODA} 2022}, pages 252--274. {SIAM}, 2022.

\bibitem[ANR95]{AzarNR95}
Yossi Azar, Joseph Naor, and Raphael Rom.
\newblock The competitiveness of on-line assignments.
\newblock {\em J. Algorithms}, 18(2):221--237, 1995.

\bibitem[AZM18]{agrawal2018proportional}
Shipra Agrawal, Morteza Zadimoghaddam, and Vahab Mirrokni.
\newblock Proportional allocation: Simple, distributed, and diverse matching
  with high entropy.
\newblock In {\em International Conference on Machine Learning}, pages 99--108.
  PMLR, 2018.

\bibitem[BGGJ22]{BanerjeeGGJ22}
Siddhartha Banerjee, Vasilis Gkatzelis, Artur Gorokh, and Billy Jin.
\newblock Online nash social welfare maximization with predictions.
\newblock In {\em Proceedings of the 2022 {ACM-SIAM} Symposium on Discrete
  Algorithms, {SODA} 2022}, pages 1--19. {SIAM}, 2022.

\bibitem[BKM22]{BarmanKM22}
Siddharth Barman, Arindam Khan, and Arnab Maiti.
\newblock Universal and tight online algorithms for generalized-mean welfare.
\newblock In {\em Thirty-Sixth {AAAI} Conference on Artificial Intelligence},
  pages 4793--4800. {AAAI} Press, 2022.

\bibitem[BMRS20]{BamasMRS20}
{\'{E}}tienne Bamas, Andreas Maggiori, Lars Rohwedder, and Ola Svensson.
\newblock Learning augmented energy minimization via speed scaling.
\newblock In {\em Advances in Neural Information Processing Systems 33, NeurIPS
  2020}, 2020.

\bibitem[Cal65]{callebaut1965generalization}
DK~Callebaut.
\newblock Generalization of the cauchy-schwarz inequality.
\newblock {\em Journal of mathematical analysis and applications},
  12(3):491--494, 1965.

\bibitem[Car08]{Caragiannis08}
Ioannis Caragiannis.
\newblock Better bounds for online load balancing on unrelated machines.
\newblock In {\em Proceedings of the Nineteenth Annual {ACM-SIAM} Symposium on
  Discrete Algorithms, {SODA} 2008}, pages 972--981. {SIAM}, 2008.

\bibitem[CI21]{ChenI21}
Justin~Y. Chen and Piotr Indyk.
\newblock Online bipartite matching with predicted degrees.
\newblock {\em CoRR}, 2021.

\bibitem[HKPS22]{HajiaghayiKPS22}
MohammadTaghi Hajiaghayi, MohammadReza Khani, Debmalya Panigrahi, and Max
  Springer.
\newblock Online algorithms for the santa claus problem.
\newblock In {\em Advances in Neural Information Processing Systems 35, NeurIPS
  2022}, 2022.

\bibitem[IKQP21]{ImKQP21}
Sungjin Im, Ravi Kumar, Mahshid~Montazer Qaem, and Manish Purohit.
\newblock Non-clairvoyant scheduling with predictions.
\newblock In {\em {SPAA} '21: 33rd {ACM} Symposium on Parallelism in Algorithms
  and Architectures, Virtual Event, USA, 6-8 July, 2021}, pages 285--294.
  {ACM}, 2021.

\bibitem[KPS{\etalchar{+}}19]{KumarPSSV19}
Ravi Kumar, Manish Purohit, Aaron Schild, Zoya Svitkina, and Erik Vee.
\newblock Semi-online bipartite matching.
\newblock In {\em 10th Innovations in Theoretical Computer Science Conference,
  {ITCS} 2019}, volume 124 of {\em LIPIcs}, pages 50:1--50:20. Schloss Dagstuhl
  - Leibniz-Zentrum f{\"{u}}r Informatik, 2019.

\bibitem[LLMV20]{LattanziLMV20}
Silvio Lattanzi, Thomas Lavastida, Benjamin Moseley, and Sergei Vassilvitskii.
\newblock Online scheduling via learned weights.
\newblock In {\em Proceedings of the 2020 {ACM-SIAM} Symposium on Discrete
  Algorithms, {SODA} 2020}, pages 1859--1877. {SIAM}, 2020.

\bibitem[LMRX21a]{LavastidaMRX21a}
Thomas Lavastida, Benjamin Moseley, R.~Ravi, and Chenyang Xu.
\newblock Learnable and instance-robust predictions for online matching, flows
  and load balancing.
\newblock In {\em 29th Annual European Symposium on Algorithms, {ESA} 2021},
  volume 204 of {\em LIPIcs}, pages 59:1--59:17, 2021.

\bibitem[LMRX21b]{LavastidaMRX21b}
Thomas Lavastida, Benjamin Moseley, R.~Ravi, and Chenyang Xu.
\newblock Using predicted weights for ad delivery.
\newblock In {\em Applied and Computational Discrete Algorithms, {ACDA} 2021},
  2021.

\bibitem[LV21]{LykourisV21}
Thodoris Lykouris and Sergei Vassilvitskii.
\newblock Competitive caching with machine learned advice.
\newblock {\em J. {ACM}}, 68(4):24:1--24:25, 2021.

\bibitem[LX21]{LiX21}
Shi Li and Jiayi Xian.
\newblock Online unrelated machine load balancing with predictions revisited.
\newblock In {\em Proceedings of the 38th International Conference on Machine
  Learning, {ICML} 2021}, 2021.

\bibitem[Mil25]{milne1925note}
EA~Milne.
\newblock Note on rosseland's integral for the stellar absorption coefficient.
\newblock {\em Monthly Notices of the Royal Astronomical Society}, 85:979--984,
  1925.

\bibitem[Mit20]{Mitzenmacher20}
Michael Mitzenmacher.
\newblock Scheduling with predictions and the price of misprediction.
\newblock In {\em 11th Innovations in Theoretical Computer Science Conference,
  {ITCS} 2020}, volume 151 of {\em LIPIcs}, pages 14:1--14:18. Schloss Dagstuhl
  - Leibniz-Zentrum f{\"{u}}r Informatik, 2020.

\bibitem[MNS12]{MahdianNS12}
Mohammad Mahdian, Hamid Nazerzadeh, and Amin Saberi.
\newblock Online optimization with uncertain information.
\newblock {\em {ACM} Trans. Algorithms}, 8(1):2:1--2:29, 2012.

\bibitem[MV20]{MitzenmacherV20}
Michael Mitzenmacher and Sergei Vassilvitskii.
\newblock Algorithms with predictions.
\newblock In {\em Beyond the Worst-Case Analysis of Algorithms}, pages
  646--662. Cambridge University Press, 2020.

\bibitem[MV22]{MitzenmacherV22}
Michael Mitzenmacher and Sergei Vassilvitskii.
\newblock Algorithms with predictions.
\newblock {\em Commun. {ACM}}, 65(7):33--35, 2022.

\bibitem[PSK18]{PurohitSK18}
Manish Purohit, Zoya Svitkina, and Ravi Kumar.
\newblock Improving online algorithms via {ML} predictions.
\newblock In {\em Advances in Neural Information Processing Systems 31, NeurIPS
  2018}, 2018.

\bibitem[RS89]{rothblum1989scalings}
Uriel~G Rothblum and Hans Schneider.
\newblock Scalings of matrices which have prespecified row sums and column sums
  via optimization.
\newblock {\em Linear Algebra and its Applications}, 114:737--764, 1989.

\end{thebibliography}

\end{document}